\documentclass[acmtocl,acmnow]{acmtrans2m}

\usepackage{latexsym}
\usepackage{amssymb}

\newtheorem{theorem}{Theorem}[section]

\newtheorem{corollary}[theorem]{Corollary}
\newtheorem{proposition}[theorem]{Proposition}
\newtheorem{lemma}[theorem]{Lemma}
\newdef{definition}[theorem]{Definition}
\newdef{remark}[theorem]{Remark}

\newtheorem{claim}{Claim}
\newcommand{\ignore}[1] { }

\newcommand{\mm}[1]{\mbox{\boldmath \tiny$#1$}}

\newcommand{\A}{\mbox{$\mathcal{A}$}}

\newcommand{\cM}{\mbox{$\mathcal{M}$}}

\newcommand{\cR}{\mbox{$\mathcal{R}$}}
\newcommand{\cS}{\mbox{$\mathcal{S}$}}

\newcommand{\sA}{{\mm \A}}

\newcommand{\ttU}{\mbox{$\texttt{U}$}}
\newcommand{\ttX}{\mbox{$\texttt{X}$}}

\newcommand{\ttLift}{\mbox{$\texttt{lift-pebble}$}}
\newcommand{\ttPlace}{\mbox{$\texttt{place-pebble}$}}
\newcommand{\ttLeft}{\mbox{$\texttt{left}$}}
\newcommand{\ttRight}{\mbox{$\texttt{right}$}}
\newcommand{\ttStay}{\mbox{$\texttt{stay}$}}
\newcommand{\ttAct}{\mbox{$\texttt{act}$}}

\newcommand{\sttLift}{\mbox{$\texttt{\scriptsize lift}$}}
\newcommand{\sttPlace}{\mbox{$\texttt{\scriptsize place}$}}
\newcommand{\sttLeft}{\mbox{$\texttt{\scriptsize left}$}}
\newcommand{\sttRight}{\mbox{$\texttt{\scriptsize right}$}}
\newcommand{\sttStay}{\mbox{$\texttt{\scriptsize stay}$}}

\newcommand{\sfmin}{\mbox{$\textsf{min}$}}
\newcommand{\sfmax}{\mbox{$\textsf{max}$}}

\newcommand{\sfTrue}{\mbox{$\textsf{True}$}}
\newcommand{\sfFalse}{\mbox{$\textsf{False}$}}

\newcommand{\sffqr}{\mbox{$\textsf{fqr}$}}

\newcommand{\fD}{\mbox{$\mathfrak{D}$}}

\newcommand{\sfImage}{\mbox{$\textsf{Image}$}}

\newcommand{\RA}{\mbox{$\textrm{RA}$}}

\newcommand{\MSO}{\mbox{$\textrm{MSO}$}}
\newcommand{\FO}{\mbox{$\textrm{FO}$}}
\newcommand{\LTL}{\mbox{$\textrm{LTL}$}}

\title{Graph Reachability and Pebble Automata over Infinite Alphabets}
            
\author{Tony Tan\\
University of Edinburgh}
            
\begin{abstract} 
Let $\fD$ denote an infinite alphabet -- 
a set that consists of infinitely many symbols.
A word $w = a_0b_0a_1b_1\cdots a_nb_n$ of even length over $\fD$ 
can be viewed as a directed graph $G_w$
whose vertices are the symbols that appear in $w$,
and the edges are $(a_0,b_0),(a_1,b_1),\ldots,(a_n,b_n)$.
For a positive integer $m$, define a language $\cR_m$ such that 
a word $w=a_0b_0\cdots a_nb_n \in \cR_m$
if and only if
there is a path in the graph $G_w$ of length $\leq m$ from the vertex $a_0$ to the vertex $b_n$.

We establish the following hierarchy theorem for pebble automata over infinite alphabet.
For every positive integer $k$,
(i) there exists a $k$-pebble automaton that accepts the language $\cR_{2^k-1}$;
(ii) there is no $k$-pebble automaton that accepts the language $\cR_{2^{k+1} - 2}$.
Based on this result, we establish
a number of previously unknown relations among 
some classes of languages over infinite alphabets.
\end{abstract}
            
         
            

\begin{document}
            
\begin{bottomstuff} 
The extended abstract of this article has been published in
the proceedings of LICS 2009.
This work was done while the author was a PhD student in
Department of Computer Science, Technion -- Israel Institute of Technology.
\end{bottomstuff}
            
\maketitle

\section{Introduction}
\label{s: intro}

Logic and automata for words over finite alphabets
are relatively well understood and
recently there is a broad research activity on logic and automata
for words and trees over infinite alphabets.
Partly, the study of infinite alphabets is motivated by
the need for formal verification and
synthesis of infinite-state systems and partly,
by the search for automated reasoning techniques for XML.
There has been a significant progress in this area,
see~\cite{BjorklundS07,BojanczykDMSS11,DemriL09,KaminskiF94,NevenSV04,Segoufin06}
and this paper aims to contribute to the progress.

Roughly speaking, there are two approaches to studying languages over infinite alphabets:
logic and automata.
Below is a brief summary on both approaches.
For a more comprehensive survey,
we refer the reader to~\cite{Segoufin06}.
The study of languages over infinite alphabets
starts with the introduction of {finite-memory} automata (FMA) in~\cite{KaminskiF94},
also known as {\em register automata} (RA),
that is, automata with a finite number of registers.
From here on, we write $\RA_n$ to denote RA with $n$ registers.

The study of RA was continued and extended in~\cite{NevenSV04},
in which {\em pebble automata} (PA) were also introduced.
Each of these models has its own advantages and disadvantages.
Languages accepted by RA are closed under standard language operations:
intersection, union, concatenation, and Kleene star.
In addition, from the computational point of view,
RA are a much easier model to handle.
Their emptiness problem is decidable,
whereas the same problem for PA is not.
However, the PA languages possess a very nice logical property:
closure under {\em all} boolean operations.

Recently there is a more general model of RA introduced in~\cite{BojanczykKL11},
that builds on the idea of nominal sets.
In this model the structure for the symbols is richer.
In addition to equality test, it allows for total order and partial order tests
among the symbols.

In~\cite{Bouyer02} data words are introduced,
which are an extension of words over infinite alphabet.
Data words are words in which each position
carries both a label from a finite alphabet,
and a data value from an infinite alphabet.
The paper~\cite{BojanczykDMSS11} studies the logic for data words, and
introduced the so-called {\em data automata}.
It was shown that data automata define the logic $\exists\MSO^2(\sim,<,+1)$,
the fragment of existential monadic second order logic in which
the first order part is restricted to two variables only,
with the signatures: the data equality $\sim$, the order $<$ and the successor $+1$.
An important feature of data automata is
that their emptiness problem is decidable, even for infinite words,
but is at least as hard as reachability for Petri nets.
It was also shown that the satisfiability problem for
the three-variable first order logic is undecidable.

Another logical approach is via the so called
{\em linear temporal logic} with freeze quantifier,
introduced in~\cite{DemriLN05} and later also studied in~\cite{DemriL09}.
Intuitively, these are LTL formula equipped
with a finite number of registers to store the data values.
We denote by $\LTL_n^{\downarrow}[\texttt{X},\texttt{U}]$,
the LTL with freeze quantifier,
where $n$ denotes the number of registers and
the only temporal operators allowed are 
the neXt operator $\ttX$ and the Until operator $\ttU$.
It was shown that alternating $\RA_n$
accept all LTL$_n^{\downarrow}[\texttt{X},\texttt{U}]$ languages
and the emptiness problem for alternating $\RA_1$ is decidable.
However, the complexity is non primitive recursive.
Hence, the satisfiability problem for
LTL$_1^{\downarrow}(\texttt{X},\texttt{U})$
is decidable as well.
Adding one more register or past time operators, such as $\texttt{X}^{-1}$ or $\texttt{U}^{-1}$,
to LTL$_1^{\downarrow}(\texttt{X},\texttt{U})$
makes the satisfiability problem undecidable.
In~\cite{Lazic11} a weaker version of alternating $\RA_1$,
called safety alternating $\RA_1$, is considered,
and the emptiness problem is shown to be EXPSPACE-complete.

In this paper we continue the study of pebble automata (PA)
for strings over infinite alphabets 
introduced in~\cite{NevenSV04}.
The original PA for strings over finite alphabet was first introduced and studied in~\cite{harel}.
Essentially PA are finite state automata equipped with a finite number of pebbles, 
The pebbles are placed on or lifted from the input word
in the stack discipline~-- first in last out~-- and
are intended to mark the positions in the input word.
One pebble can only mark one position and
the most recently placed pebble serves as the head of the automaton.
The automaton moves from one state to another
depending on the equality tests among data values in the
positions currently marked by the pebbles, as well as 
the equality tests among the positions of the pebbles.

As mentioned earlier, PA languages possess a very nice logical
property: closure under {\em all} boolean operations.
Another desirable property of PA languages is, as shown in~\cite{NevenSV04},
that nondeterminism and two-way-ness do not increase the
expressive power of PA~\cite[Theorem~4.6]{NevenSV04}.
Moreover, the class of PA languages lies strictly in between
$\FO(\sim,<,+1)$ and $\MSO(\sim,<,+1)$~\cite[Theorems~4.1 and~4.2]{NevenSV04}.

Moreover, looking at the stack discipline imposed on the placement of the pebbles,
one can rightly view PA as a natural extension of $\FO(\sim,<,+1)$.
To simulate a first-order sentence of quantifier rank $k$,
a pebble automaton with $k$ pebbles suffices:
one pebble for each quantifier depth.
(See Proposition~\ref{p: FO rank k}.)

In this paper we study PA as a model of computation for the directed graph reachability problem.
To this end, we view a word of even length $w=a_0b_0 a_1b_1 \cdots a_nb_n$ 
over an infinite alphabet as a directed graph $G_w=(V_w,E_w)$ 
with the symbols that appear in
$a_0b_0a_1b_1\cdots a_nb_n$ as the vertices in $V_w$ and
$(a_0,b_0),\ldots,(a_n,b_n)$ as the edges in $E_w$. 
We say that the word $w$ induces the graph $G_w$.

We prove that for any positive integer $k$,
$k$ pebbles are sufficient for recognizing the existence of a path
of length $2^k-1$ from the vertex $a_0$ to the vertex $b_n$, but
are not sufficient for recognizing the existence of a path of
length $2^{k+1} - 2$ from the vertex $a_0$ to the vertex $b_n$.
Based on this result, we establish the following relations
among the classes of languages over infinite alphabets which were
previously unknown.
\begin{enumerate}
\item
A strict hierarchy of the PA languages based on the number of pebbles.
\item
The separation of monadic second order logic from the PA languages.
\item
The separation of one-way deterministic RA languages from PA languages.
\end{enumerate}
Some of these results settle questions left open in~\cite{NevenSV04,Segoufin06}.

Although, in general, the emptiness problem for PA is undecidable,
we believe that our study may contribute to the technical aspect
of reasoning on classes of languages with decidable properties.
For example, in Section~\ref{s: weak PA} a similar technique is used
to obtain separation result for LTL$^{\downarrow}_1[\ttX,\ttU]$ languages,
a class of languages with decidable satisfiability problem.

\paragraph*{Related work}
A weaker version of PA, called {\em top-view weak} PA 
was introduced and studied in~\cite{Tan10}, where
it was also shown that the emptiness problem is decidable.
The results in this paper are not implied from that paper,
as here the main concern is separation results.
In fact, some of the separation results here
also hold for the model in~\cite{Tan10}.

There is also an analogy between our result with the classical
first-order quantifier lower bounds for directed graph $(s,t)$-reachability
which states the following:
There is a first order sentence of quantifier rank $k$
to express the existence of a path of length $\leq m$ from the source node $s$ to the target node $t$ 
if and only if $m \leq 2^k$. See, for example,~\cite{Turan84}.

As far as we can see, our result is actually a tighter version of the classical result
for first-order logic.
It is tighter because PA is shown to be stronger than
first-order logic (Proposition~\ref{p: FO rank k}).
In particular pebble automata do have states, thus, enjoy the usual benefits
associated with automata, like counting the number of edges, or the number of neighbours
up to $\leq m$, $\geq m$, or mod $m$, for an arbitrary but fixed positive integer $m$,
without increasing the number of pebbles.

Other related results are those established in~\cite{AjtaiF90,FaginSV95,Schwentick96}. 
To the best of our knowledge, those
results have no connection with the result in this paper.
In~\cite{AjtaiF90} it is established that $(s,t)$-reachability in
directed graph is not in monadic NP\footnote{Monadic NP is a complexity theoretic name
for existential monadic second order logic.}, 
while in~\cite{FaginSV95,Schwentick96} it is established that 
undirected graph connectivity is not in monadic NP. However, no lower bound on first-order
quantifier rank is established. 

\paragraph*{Organization}
This paper is organized as follows. 
In Section~\ref{s: model} we review 
the monadic second-order logic $\MSO(\sim,<,+1)$ 
and pebble automata (PA) for words over infinite alphabet. 
Section~\ref{s: graph} is the core of the
paper in which we present our main results. In
Section~\ref{s: weak PA} we discuss how to adjust our
results and proofs presented in Section~\ref{s: graph} to a weaker
version of PA, called weak PA,
%which is also related to the logic LTL$^{\downarrow}_1(\ttX,\ttU)$.
 whose relation to
 the logic LTL$^{\downarrow}_1(\ttX,\ttU)$
 is presented in Section~\ref{s: ltl}.

\section{Models of computations}
\label{s: model}

In Section~\ref{ss: pa} we recall the definition of alternating pebble
automata from~\cite{NevenSV04}, and in Section~\ref{ss: logic} 
a logic for languages over infinite alphabets.

We shall use the following notation: $\fD$ is a fixed infinite
alphabet not containing the left-end marker $\triangleleft$ or the
right-end marker  $\triangleright$. The input word to an
automaton is of the form $\triangleleft w \triangleright$, where
$w \in \fD^\ast$. Symbols of $\fD$ are denoted by lower case
letters~$a$,~$b$,~$c$, etc., possibly indexed, and 
words over $\fD$ by lower case letters~$u$,~$v$,~$w$, etc.,
possibly indexed. 

\subsection{Pebble automata}
\label{ss: pa}

\begin{definition}
\label{d: pa}
(See~\cite[Definition~2.3]{NevenSV04})
A {\em two-way alternating} $k$-{\em pebble automaton}, (in short
$k$-PA) is a system $\A = \langle Q, q_0, F, \mu, U\rangle$ whose
components are defined as follows.
\begin{enumerate}
\item 
$Q$, $q_0 \in Q$ and $F\subseteq Q$ are a finite set of 
{\em states}, the {\em initial state}, and the set of {\em final
states}, respectively; 
\item 
$U \subseteq Q - F$ is the set of
{\em universal} states; and 
\item 
$\mu$ is a finite set of
transitions of the form $\alpha \rightarrow \beta$ such that
\begin{itemize}
\item 
$\alpha$ is of the form $(i,P,V,q)$,
where $i \in \{1,\ldots,k\}$, $P,V \subseteq\{i+1,\ldots,k\}$,
$q \in Q$ 
and 
\item $\beta$ is of the form $(q,\ttAct)$,
where $q \in Q$ and
\[
\ttAct\in \{ \ttLeft,\ttRight,\ttStay,\ttPlace,\ttLift \}.
\]
\end{itemize}
The intuitive meaning of $P$ and $V$ in $(i,P,V,q)$
is that $P$ denotes the set of pebbles 
that occupy the same position as pebble~$i$,
while $V$ the set of pebbles
that read the same symbol as pebble~$i$.
A more precise explanation can be found below. 
\end{enumerate}
\end{definition}

Given a word $w = a_1\cdots a_n \in \fD^\ast$, 
a {\em configuration of $\A$ on $\triangleleft w \triangleright$} is a
triple $\gamma = [i,q,\theta]$, where $i \in \{1,\ldots,k\}$, $q \in Q$ and
$\theta : \{i,i+1,\ldots,k\} \rightarrow \{ 0,1,\ldots,n,n + 1 \}$. 
The function $\theta$ defines the position of the pebbles and
is called the {\em pebble assignment} of $\gamma$. 
The symbols in the positions $0$ and $n + 1$ are $\triangleleft$ and
$\triangleright$, respectively.
That is, we count the leftmost position in $w$ as position $1$.

The {\em initial configuration} of $\A$ on $w$ 
is $\gamma_0 = [k,q_0,\theta_0]$,
where $\theta_0(k) = 0$ is the {\em initial pebble assignment}. 
A configuration $[i,q,\theta]$ with $q \in F$ 
is called an {\em accepting configuration}.
A transition $(i,P,V,p) \rightarrow \beta$ 
{\em applies to a configuration $[j,q,\theta]$}, if
\begin{enumerate}
\item[$(1)$] 
$i = j$ and $p = q$, 
\item[$(2)$] 
$P = \{ l > i \mid \theta(l) = \theta(i) \}$, and
\item[$(3)$] 
$V = \{ l > i \mid a_{\theta(l)} = a_{\theta(i)} \}$.
\end{enumerate}
We define the transition relation $\vdash_{\sA}$ on $\triangleleft w \triangleright$ as follows:
$[i,q,\theta]\vdash_{\sA,w} [i^\prime,q^\prime,\theta^\prime]$, if
there is a transition $\alpha \rightarrow (p,\ttAct) \in \mu$ that
applies to $[i,q,\theta]$ such that $q^\prime = p$, 
for all $j > i$, $\theta^\prime(j)=\theta(j)$, and
\begin{itemize}
\item[-] 
if $\ttAct = \ttLeft$, then $i^\prime=i$ and
$\theta^\prime(i)=\theta(i)-1$, 
\item[-] 
if $\ttAct = \ttRight$,
then $i^\prime=i$ and $\theta^\prime(i)=\theta(i)+1$, 
\item[-] 
if $\ttAct = \ttStay$,
then $i^\prime=i$ and $\theta^\prime(i)=\theta(i)$, 
\item[-] 
if $\ttAct = \ttLift$, then $i^\prime=i+1$, 
\item[-] 
if $\ttAct = \ttPlace$, then $i^\prime=i-1$, $\theta^\prime(i-1) = 0$ and
$\theta^\prime(i)=\theta(i)$.
\end{itemize}
As usual, we denote the reflexive, transitive closure of $\vdash_{\sA,w}$ by
$\vdash^\ast_{\sA,w}$. When the automaton $\A$ and the word $w$ are clear from the
context, we shall omit the subscripts $\A$ and $w$. 
For $1\leq i \leq k$, an $i$-configuration is a configuration of the form
$[i,q,\theta]$, that is, when the head pebble is pebble~$i$.

\begin{remark}
\label{r: data PA}
Here we define PA as a model of computation for
languages over infinite alphabet. Another option is to define PA
as a model of computation for data words. A data word is a finite
sequence of $\Sigma\times\fD$, where $\Sigma$ is a finite alphabet
of labels. There is only a slight technical difference between the
two models. Every data word can be viewed as a word over infinite
alphabet in which every odd position contains a constant symbol.
In the context of our paper, we ignore the finite labels,
thus, Definition~\ref{d: pa} is more convenient.
\end{remark}

We now define how pebble automata accept words.
Let $\gamma = [i,q,\theta]$ be a configuration of a PA $\A$ on a word $w$.
We say that $\gamma$ {\em leads to acceptance},
if and only if either $q\in F$, or the following conditions hold.
\begin{itemize}
\item 
if $q \in U$, then for all configurations $\gamma'$ such that $\gamma\vdash\gamma'$,
$\gamma'$ leads to acceptance.
\item 
if $q \notin F\cup U$, then there is at least one
configuration $\gamma'$ such that $\gamma \vdash \gamma'$ and
$\gamma'$ leads to acceptance.
\end{itemize}
A word $w \in \fD^\ast$ is accepted by $\A$,
if the initial configuration $\gamma_0$ leads to acceptance.
The language $L(\A)$
consists of all data words accepted by $\A$.

The automaton $\A$ is {\em nondeterministic}, if the set
$U=\emptyset$, and it is {\em deterministic}, if 
for each configuration, there is exactly
one transition that applies. If $\ttAct \in
\{\ttRight, \ttLift,\ttPlace\}$ for all transitions, then the
automaton is {\em one-way}. It turns out that PA languages are
quite robust. Namely, alternation and two-wayness do not increase
the expressive power to one-way deterministic PA, 
as stated in Theorem~\ref{t: equivalence} below.

\begin{remark}
In~\cite{NevenSV04} the model defined above is called {\em strong} PA.
A weaker model in which the new pebble is placed at the position
of the head pebble, is referred to as {\em weak} PA.
Obviously for two-way PA, strong and weak PA are equivalent.
However, for one-way PA, strong PA is indeed stronger than weak PA.
We will postpone our discussion of weak PA until Section~\ref{s: weak PA}.
\end{remark}

\begin{theorem}
\label{t: equivalence}
For each $k \geq 1$, two-way alternating $k$-PA
and one-way deterministic $k$-PA have the same recognition power.
\end{theorem}

The proof of Theorem~\ref{t: equivalence} is a straightforward
adaption of the classical proof of the equivalence
between the expressive power of alternating two-way finite state automata
and deterministic one-way finite state automata~\cite{LadnerLS84}.
For this reason, we omit the proofs.

The main idea is that when pebble~$i$ is the head pebble,
due to the stack discipline imposed on placing the pebbles,
all the other pebbles (pebbles~$i+1,\ldots,k$) are fixed on their positions.
Hence the transitions of pebble~$i$,
which are of the form $(i,P,V,q) \to (p,\ttAct)$,
can be viewed as transitions over the finite alphabet 
$(P,V)\in 2^{\{i+1,\ldots,k\}}\times 2^{\{i+1,\ldots,k\}}$.
Thus, the idea in~\cite{LadnerLS84} can be adapted to PA in a straightforward manner.
The details are available as a technical report in~\cite{Tan09Alternating}.
In view of this equivalence, we will always assume
that the pebble automata under consideration are deterministic and one-way.

Next, we define the hierarchy of languages accepted by PA. For $k
\geq 1$, $\textrm{PA}_k$ is the set of all languages accepted by
$k$-PA, and $\textrm{PA}$ is the set of all languages accepted by pebble automata. 
That is,
\[
\textrm{PA} = \bigcup_{k \geq 1} \textrm{PA}_k .
\]

\subsection{Logic}
\label{ss: logic}

Formally, a word $w=a_1\cdots a_n$ is represented
by the logical structure with domain $\{1,\ldots,n\}$;
the natural ordering $<$ on the domain with its induced successor $+1$;
and the equivalence relation $\sim$ on the domain $\{1,\ldots,n\}$,
where $i \sim j$ whenever $a_i = a_j$.

The atomic formulas in this logic are of the form
$x < y$, $y=x+1$, $x\sim y$.
The first-order logic $\FO(\sim,<,+1)$ is obtained by closing
the atomic formulas under the propositional connectives and
first-order quantification over $\{1,\ldots,n\}$.
The second-order logic $\MSO(\sim,<,+1)$ is obtained by adding quantification over unary
predicates on $\{1,\ldots,n\}$. A sentence $\varphi$ defines the
set of words
$$
L(\varphi) = \{w \mid w \models \varphi\}.
$$
If $L=L(\varphi)$ for some sentence $\varphi$, then we say that
the sentence $\varphi$ {\em expresses} the language $L$.

We use the same notations $\FO(\sim,<,+1)$ and $\MSO(\sim,<,+1)$ to denote the
languages expressible by sentences in $\FO(\sim,<,+1)$ and $\MSO(\sim,<,+1)$,
respectively. That is,
\[
\FO(\sim,<,+1) = \{L(\varphi) \mid \varphi \ \mbox{is an} \ \FO(\sim,<,+1) \ \mbox{sentence}\}
\]
and
\[
\MSO(\sim,<,+1) = \{L(\varphi) \mid \varphi \ \mbox{is an} \ \MSO(\sim,<,+1) \ \mbox{sentence}\}.
\]

\begin{proposition}
\label{p: FO rank k} {\em (See also~\cite[Theorem~4.1]{NevenSV04})}
If $\varphi \in \FO(\sim,<,+1)$ is a sentence with quantifier rank $k$, 
then $L(\varphi)\in \textrm{PA}_k$.
\end{proposition}
\begin{proof} (Sketch)
First, it is straightforward that
languages accepted by two-way alternating $k$-PA are closed under Boolean operations.
By Theorem~\ref{t: equivalence},
two-way alternating and one-way deterministic $k$-PA are equivalent. 
Thus, the class $\textrm{PA}_k$ is closed under Boolean operations.
Therefore, it is sufficient to prove Proposition~\ref{p: FO rank k}
when the formula $\varphi$ is of the form $Q x_k \psi(x_k)$, 
where $Q \in \{\forall,\exists\}$ and
$\psi(x_k)$ is a formula of quantifier rank $k-1$.

The proof is by straightforward induction on $k$. A $k$-PA $\A$
iterates pebble~$k$ through all possible positions in the input
word $w$. On each iteration, the automaton $\A$ recursively calls
a $(k-1)$-PA $\A'$ that accepts the language $L(\psi(x_k))$,
treating the position of pebble~$k$ as the assignment value for $x_k$.

The transition in the PA $\A'$ can test the atomic formula $x=y$ and $x\sim y$;
while at the same time remembering in its states
the order of the pebbles.
The word $w$ is accepted by $\A$, if the following holds.
\begin{itemize}
\item
If $Q$ is $\forall$, then
$\A$ accepts $w$ if and only if
$\A'$ accepts on all iterations.
\item
If $Q$ is $\exists$, then
$\A$ accepts $w$ if and only if $\A'$ accepts on at least one
iteration.
\end{itemize}
This completes the sketch of our proof of Proposition~\ref{p: FO rank k}.
\end{proof}

We end this section with Theorem~\ref{t: PA <= MSO} below which
states that a language accepted by pebble automaton can be
expressed by an $\MSO(\sim,<,+1)$ sentence.

\begin{theorem}
\label{t: PA <= MSO} {\em (\cite[Theorem 4.2]{NevenSV04})}
For every PA $\A$, there exists an $\MSO(\sim,<,+1)$ sentence $\varphi_{\sA}$
such that $L(\A) = L(\varphi_{\sA})$.
\end{theorem}
%Actually Theorem~4.2 in~\cite{NevenSV04} states a stronger result:
%all languages accepted by alternating two-way PA can be expressed
%by an MSO$^*$ sentence.
%However, we will postpone our discussion of alternating PA
%to Section~\ref{s: conclusion}.

\section{Words of $\fD^*$ as Graphs} 
\label{s: graph}

This section contains the main results in this paper:
\begin{enumerate}
\item
The strict hierarchy of PA languages based on the number of pebbles.
\item
The separation of $\MSO(\sim,<,+1)$ from PA languages.
\item
The separation of one-way deterministic RA languages from PA languages.
\end{enumerate}
All three results share one common idea: 
We view a word of even length as a directed graph.
Recall that $\fD$ is an infinite alphabet, and that
we always denote the symbols in $\fD$ by the lower case letters $a,b,c,\ldots$,
possibly indexed.

We consider directed graphs in which the vertices come from $\fD$.
A word $w= a_0 b_0 \cdots a_n b_n \in \fD^*$ of even length 
{\em induces} a directed graph $G_w = (V_w,E_w)$, 
where $V_w$ is the set of symbols
that appear in $w$, that is, $V_w = \{a : a \ \mbox{appears in} \ w\}$, 
and the set of edges is $E_w = \{(a_0,b_0),\ldots,(a_n,b_n)\}$. 
We also write $s_w = a_0$ and $t_w=b_n$ to denote the first and the last
symbol in $w$, respectively. 
For convenience, we consider only the words $w$ in which 
$s_w$ and $t_w$ occur only once.

As an example, we take the following word
$w = ab \ bc \ bd \ cd \ 
ce \ de \ ef \ eg$.
Then $s_w = a$ and $t_w = g$.
The graph induced by $w$ is the $G_w = (V_w,E_w)$,
where $V_w = \{a,b,c,d,e,f,g\}$ and
$E_w = \{(a,b),(b,c),(b,d),(c,d),(c,e),
(d,e),(e,f),(e,g)\}$,
as illustrated in the picture below.

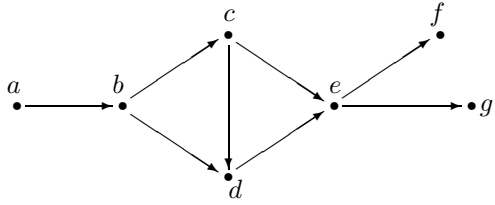
\begin{picture}(300,120)(-180,-60)

\put(-100,0){\circle*{3}}
\put(-104,5){$a$}
\put(-97,0){\vector(1,0){33}}

\put(-60,0){\circle*{3}}
\put(-64,5){$b$}
\put(-57,3){\vector(3,2){33}}
\put(-57,-3){\vector(3,-2){33}}

\put(-20,27){\circle*{3}}
\put(-22,32){$c$}
\put(-20,23){\vector(0,-1){47}}

\put(-20,-27){\circle*{3}}
\put(-20,-35){$d$}

\put(-17,24){\vector(3,-2){33}}

\put(-17,-24){\vector(3,2){33}}

\put(20,0){\circle*{3}}
\put(18,5){$e$}
\put(23,0){\vector(1,0){45}}
\put(60,27){\circle*{3}}
\put(56,32){$f$}

\put(23,3){\vector(3,2){33}}
\put(72,0){\circle*{3}}
\put(75,-2){$g$}

\end{picture}

We need the following basic graph terminology. Let $a$ and $b$ be
vertices in a graph $G$. A {\em path} of length $m$ from $a$ to
$b$ is a sequence of $m$ edges
$(a_{i_1},b_{i_1}),\ldots,(a_{i_m},b_{i_m})$ in $G$ such that
$a_{i_1}=a$, $b_{i_m}=b$ and for each $j=1,\ldots,m-1$,
$b_{i_j}=a_{i_{j+1}}$. The {\em distance} from $a$ to $b$,
denoted by $d_G(a,b)$, is the
length of the shortest path from $a$ to $b$ in $G$. If there is no
path from $a$ to $b$ in $G$, then we set $d_G(a,b)=\infty$.

We now define the following reachability languages. For $m \geq 1$,
\[
\cR_{m} = \{ w \mid  d_{G_w}(s_w,t_w) \leq m \}
\]
and
\[
\cR = \bigcup_{m=1,2,\ldots} \cR_m.
\]
Here we should remark that 
since we consider only the words $w$ in which 
$s_w$ and $t_w$ occur only once,
the language $\cR_1$ consists of words of length 2 with different symbols.

\begin{proposition}
\label{p: savitch}
For each $k=2,3,\ldots$, $\cR_{2^k-1} \in \textrm{PA}_k$.
\end{proposition}

The proof of this proposition is an implementation of
Savitch's algorithm~\cite{Savitch70} for ($s$-$t$)-reachability
by pebble automata.
It can be found in Subsection~\ref{ss: proof: savitch}.

Lemma~\ref{l: R_n_k not in PA_k} below is the backbone of most of
the results presented in this paper. 
For each $i=0,1,2,\ldots$,
we define $n_i = 2^{i+1} - 2$.
An equivalent recursive definition is $n_0 = 0$, and
$n_{i+1} = 2 n_{i} + 2$, for $i \geq 1$.

\begin{lemma}
\label{l: R_n_k not in PA_k}
For every $k$-pebble automaton $\A$, where $k \geq 1$,
there exist a word $w\in \cR_{n_k}$ and $\overline{w}\notin \cR$
such that either $\A$ accepts both $w$ and $\overline{w}$,
or $\A$ rejects both $w$ and $\overline{w}$.
\end{lemma}

The proof of Lemma~\ref{l: R_n_k not in PA_k} is rather long and technical.
We present it in Subsections~\ref{ss: proof} and ~\ref{ss: proof subclaim}.
Meanwhile we discuss a number of consequences of this lemma.
Corollary~\ref{c: R_n_k not in PA_k} below immediately follows from the lemma.
\begin{corollary}
\label{c: R_n_k not in PA_k} 
$\cR_{n_k} \notin \textrm{PA}_k$.
\end{corollary}

\begin{corollary}
\label{c: R not in PA}
$\cR \notin \textrm{PA}$.
\end{corollary}
\begin{proof}
Assume to the contrary that $\cR = L(\A)$ for a $k$-PA $\A$.
Then, by Lemma~\ref{l: R_n_k not in PA_k},
there exists a word $w\in \cR_{n_k}$ and $\overline{w}\notin \cR$
such that either $\A$ accepts both $w$ and $\overline{w}$,
or $\A$ rejects both $w$ and $\overline{w}$.
Both yield a contradiction to the assumption that $\cR = L(\A)$.
\end{proof}

The following theorem establishes the proper hierarchy of the PA
languages.

\begin{theorem}
\label{t: PA_k < PA_k+1}
For each $k=2,\ldots$, $\textrm{PA}_k \subsetneq \textrm{PA}_{k+1}$.
\end{theorem}
\begin{proof}
We contend that $\cR_{2^{k+1}-1} \in \textrm{PA}_{k+1} - \textrm{PA}_k$,
for each $k=2,\ldots,3$.
That $\cR_{2^{k+1}-1} \in \textrm{PA}_{k+1}$ follows from Proposition~\ref{p: savitch}.
That $\cR_{2^{k+1}-1} \notin \textrm{PA}_{k}$ follows from the fact that
$n_k = 2^{k+1}-2 < 2^{k+1}-1$ and Lemma~\ref{l: R_n_k not in PA_k}.
\end{proof}

Another consequence of Corollary~\ref{c: R not in PA} is that
the inclusion  of $\textrm{PA}$ in $\MSO(\sim,<,+1)$ obtained in
Theorem~\ref{t: PA <= MSO} is proper.

\begin{theorem}
\label{t: PA < MSO}
$\textrm{PA} \subsetneq \MSO(\sim,<,+1)$.
\end{theorem}
\begin{proof}
Without loss of generality, we may assume that $\MSO(\sim,<,+1)$ contains
two constant symbols, $\sfmin$ and $\sfmax$, which denote minimum
and the maximum elements of the domain, respectively. For a word
$w=a_1\cdots a_n$, the minimum and the maximum elements are $1$
and $n$, respectively, and not $0$ and $n+1$ which are reserved
for the end-markers $\triangleleft$ and $\triangleright$.

The language $\cR$ can be expressed in $\MSO(\sim,<,+1)$ as
follows. There exist unary predicates $S_{odd}$ and $P$ such that
either 
\begin{itemize}
\item
$\sfmin +1 = \sfmax \wedge \sfmin \nsim \sfmax$ (to capture $\cR_1$),
\end{itemize}
or the following holds.
\begin{itemize}
\item
For all $x$, if $x\neq \sfmin$, then $x\nsim \sfmin$.
\\
(This is to take care our assumption that the first symbol appears only once.)
\item
For all $x$, if $x\neq \sfmax$, then $x\nsim \sfmax$.
\\
(This is to take care our assumption that the last symbol appears only once.)
\item 
$S_{odd}$ is the set of all odd elements in the domain where
$\sfmin \in S_{odd}$ and $\sfmax \not\in S_{odd}$.
\item
The predicate $P$ satisfies the conjunction of the following
$\FO(\sim,<,+1)$ sentences:
\begin{itemize}
\item
$P \subseteq S_{odd}$ and $\sfmin \in P$ and $\sfmax-1 \in P$,
\item
for all $x \in P - \{\sfmax-1\}$, there exists exactly
one $y \in P$ such that $x+1\sim y$, and
\item
for all $x \in P - \{\sfmin\}$, there exists exactly one $y \in P$ such
that $y+1 \sim x$.
\end{itemize}
\end{itemize}
Now, the theorem follows from Corollary~\ref{c: R not in PA}.
\end{proof}

\begin{remark}
Combining Theorems~\ref{t: equivalence} and~\ref{t: PA < MSO},
we obtain that $\MSO(\sim,<,+1)$ is stronger than two-way alternating PA.
This settles a question left open in~\cite{NevenSV04} whether
$\MSO(\sim,<,+1)$ is strictly stronger than two-way alternating PA.
\end{remark}

Next, we define a restricted version of the reachability languages.
For a positive integer $m \geq 1$, the language $\cR^+_m$ consists
of all words of the form
$$
c_0c_1 \underbrace{\cdots}_{u_1} c_1c_2
\underbrace{\cdots}_{u_2} c_2c_3
\cdots\cdots\cdots\cdots \cdots
c_{m-3} c_{m-2} \underbrace{\cdots}_{u_{m-2}} c_{m-2} c_{m-1}
\underbrace{\cdots}_{u_{m-1}} c_{m-1} c_m
$$
where for each $i\in\{0,\ldots,m-1\}$,
the symbol $c_i$ does not appear in $u_i$ and
$c_i \neq c_{i+1}$.
The language $\cR^+$ is defined as
$$
\cR^{+}  =  \bigcup_{m=1,2,\ldots} \cR^{+}_m.
$$

\begin{remark}
\label{r: from R to R+}
Actually, in the proof of Lemma~\ref{l: R_n_k not in PA_k}
we show that for every $k$-PA $\A$, 
there exist a word $w\in \cR_{n_k}^+$ and $\overline{w}\notin \cR^+$
such that either $\A$ accepts both $w$ and $\overline{w}$,
or $\A$ rejects both $w$ and $\overline{w}$.
Therefore, $\cR^+ \not\in \textrm{PA}$.
\end{remark}

The following theorem answers a question left open in~\cite{NevenSV04,Segoufin06}:
Can one-way deterministic FMA be simulated by pebble automata?
(We refer the reader to~\cite[Definition~1]{KaminskiF94}
for the formal definition of FMA.)

\begin{theorem}
\label{t: 1FMA notin PA}
The language $\cR^+$ is accepted by
one-way deterministic FMA,
but is not accepted by pebble automata.
\end{theorem}
\begin{proof}
Note that $\cR^{+}$ is accepted by a one-way deterministic FMA 
with two registers.\footnote{Here we use the definition of FMA
as in~\cite{KaminskiF94}.
If we use the definition of RA as in~\cite{Segoufin06,DemriL09},
then one register is sufficient to accept $\cR^+$.}
On input word $w = c_0c_1 \cdots c_{n-1}c_n$,
the automaton stores $c_1$ in the first register and then moves right
(using the second register to scan the input symbols) until it
finds a symbol $c_{i} = c_1$.
If it finds one, then it stores $c_{i+1}$ in the first register and
moves right again until it finds another symbol $c_{i'}=c_{i+1}$.
It repeats the process until either of the following holds.
\begin{itemize}
\item
The symbol in the second last position $c_{n-1}$ 
is the same as the content of the first register, or,
\item
it cannot find a symbol currently stored in the first register.
\end{itemize}
In the former case, the automaton accepts the input word $w$, and in
the latter case it rejects.
By Remark~\ref{r: from R to R+}, the language $\cR^+$ is not a PA language.
This proves Theorem~\ref{t: 1FMA notin PA}.
\end{proof}

\subsection{Proof of Proposition~\ref{p: savitch}}
\label{ss: proof: savitch}

In this subsection we prove Proposition~\ref{p: savitch}.
Before we proceed with the proof,
we remark that when processing an input word $w$, an
automaton $\A$ can remember in its state whether 
a pebble is currently at an odd- or even-numbered position in $w$.
Moreover, we always denote the input
word $w$ by $a_0b_0\cdots a_nb_n$ -- 
that is, we denote the symbols on the odd positions by $a_i$'s 
and the symbols on the even position by $b_i$'s.
We can also assume that
the automaton always rejects words of odd length.

We are going to construct a $k$-PA $\A$ that accepts $\cR_{2^{k}-1}$.
Essentially the automaton $\A$ consists of the following subautomata.
\begin{itemize}
\item
An $i$-PA $\A_i^{j,j'}$, for each $i\in \{1,\ldots,k-1\}$ and $j,j' \in \{i+1,\ldots,k\}$.
\\
The purpose of each automaton $\A_i^{j,j'}$ is 
to detect the existence of a path $\leq 2^{i}-1$ 
from the vertex seen by pebble~$j$
to the vertex seen by pebble~$j'$.
\item
An $i$-PA $\A_i^{*,j}$, for each $i\in \{1,\ldots,k-1\}$ and $j \in \{i+1,\ldots,k\}$.
\\
The purpose of each automaton $\A_i^{*,j}$ is 
to detect the existence of a path $\leq 2^{i}-1$ 
from the vertex $s_w$ to the vertex seen by pebble~$j$.
\item
An $i$-PA $\A_i^{j,*}$, for each $i\in \{1,\ldots,k-1\}$ and $j \in \{i+1,\ldots,k\}$.
\\
The purpose of the automaton $\A_i^{j,*}$ is 
to detect the existence of a path $\leq 2^{i}-1$ 
from vertex seen by pebble~$j$ to the vertex $t_w$.
\end{itemize}
We are going to show how to construct those subautomata 
$\A_i^{j,j'}$, $\A_i^{j,*}$ and $\A_i^{*,j}$
by induction on $i$.

The basis is $i=1$.
The construction of $\A_1^{j,j'}$, $\A_1^{j,*}$ and $\A_1^{*,j}$ is as follows.
\begin{itemize}
\item
The automaton $\A_1^{j,j'}$ performs the following.
\begin{enumerate}
\item
It checks whether the symbols seen by pebbles~$j$ and $j'$ are the same,
which means that there is a path of length $0$ 
from the vertex seen by pebble~$j$ to the vertex seen by pebble~$j'$.
\item
Otherwise, it iterates pebble~1 on every odd position in $w$ checking whether
there exists an index $l$ such that 
$a_l$ is the same symbol seen by pebble~$j$.
If there is,
it moves to the right one step to read $b_l$ and 
checks whether it is the same symbol seen by pebble~$j'$.
This means that there is a path of length $1$ 
from the vertex seen by pebble~$j$ to the vertex seen by pebble~$j'$.
\end{enumerate}
\item
The automaton $\A_1^{*,j}$ simply puts pebble~1 on the second position of $w$ 
to read $b_0$ and checks whether it is the same symbol seen by pebble~$j$.
(Here we use the assumption that $s_w$ occurs only once in $w$,
which implies that there cannot be a path of length 0 in this case.)
\item
The automaton $\A_1^{j,*}$ simply puts pebble~1 on the second last position of $w$ 
to read $a_n$ and checks whether it is the same symbol seen by pebble~$j$.
(Here we use the assumption that $t_w$ occurs only once in $w$,
which implies that there cannot be a path of length 0 in this case.)
\end{itemize}

For the induction step,
we describe the construction of the automata $\A_i^{j,j'}$, $\A_i^{j,*}$ and $\A_i^{*,j}$
as follows.
\begin{itemize}
\item
The automaton $\A_i^{j,j'}$ performs the following.
It iterates pebble~$i$ on each position in the input word $w$.
\begin{enumerate}
\item
When pebble~$i$ is on the odd position reading the symbol $a_l$,
it invokes the automaton $\A_{i-1}^{j,i}$
to check whether there exists a path of length $\leq 2^{i-1}-1$
from the vertex seen by pebble~$j$ to the vertex $a_l$.
\item
If there is such a path,
it moves pebble~$i$ one step to the right reading the symbol $b_l$.
It then invokes the automaton $\A_{i-1}^{i,j'}$
to check whether there exists a path of length $\leq 2^{i-1}-1$
from the vertex $b_l$ to the vertex seen by pebble~$j'$.
\end{enumerate}
Now there exists a path of length $\leq 2^{i}-1$ from the vertex seen by pebble~$j$
to the vertex seen by pebble~$j'$ if and only if
there exists an index $l$ such that
(i) there exists a path of length $\leq 2^{i-1}-1$
from the vertex seen by pebble~$j$ to the vertex $a_l$, and
(ii) there exists a path of length $\leq 2^{i-1}-1$
from the vertex $b_l$ to the vertex seen by pebble~$j'$.
This implies the correctness of our construction of $\A_i^{j,j'}$.
\item
The automaton $\A_i^{*,j}$ performs the following.
It iterates pebble~$i$ on each position in the input word $w$.
\begin{enumerate}
\item
When pebble~$i$ is on the odd position reading the symbol $a_l$,
it invokes the automaton $\A_{i-1}^{*,i}$
to check whether there exists a path of length $\leq 2^{i-1}-1$
from the vertex $s_w$ to the vertex $a_l$.
\item
If there is such a path,
it moves pebble~$i$ one step to the right reading the symbol $b_l$.
It then invokes the automaton $\A_{i-1}^{i,j}$
to check whether there exists a path of length $\leq 2^{i-1}-1$
from the vertex $b_l$ to the vertex seen by pebble~$j'$.
\end{enumerate}
It follows immediately that $\A_i^{*,j}$ checks
the existence of a path $\leq 2^i-1$ 
from the vertex $s_w$ to the vertex seen by pebble~$j$.

\item
The automaton $\A_i^{j,*}$ performs the following.
It iterates pebble~$i$ on each position in the input word $w$.
\begin{enumerate}
\item
When pebble~$i$ is on the odd position reading the symbol $a_l$,
it invokes the automaton $\A_{i-1}^{j,i}$
to check whether there exists a path of length $\leq 2^{i-1}-1$
from the vertex seen by pebble~$j$ to the vertex $a_l$.
\item
If there is such a path,
it moves pebble~$i$ one step to the right reading the symbol $b_l$.
It then invokes the automaton $\A_{i-1}^{i,*}$
to check whether there exists a path of length $\leq 2^{i-1}-1$
from the vertex $b_l$ to the vertex $t_w$.
\end{enumerate}
It follows immediately that $\A_i^{*,j}$ checks
the existence of a path $\leq 2^i-1$ 
from the vertex seen by pebble~$j$ to the vertex $t_w$.
\end{itemize}

Now the automaton $\A$ performs the following.
It iterates pebble~$k$ on each position in the input word $w$.
\begin{enumerate}
\item
When pebble~$k$ is on the odd position reading the symbol $a_l$,
it invokes the automaton $\A_{k-1}^{*,k}$
to check whether there exists a path of length $\leq 2^{k-1}-1$
from the vertex $s_w$ to the vertex $a_l$.
\item
If there is such a path,
it moves pebble~$k$ one step to the right reading the symbol $b_l$.
It then invokes the automaton $\A_{k-1}^{k,*}$
to check whether there exists a path of length $\leq 2^{k-1}-1$
from the vertex $b_l$ to the vertex $t_w$.
\end{enumerate}
Hence, $\A$ is the desired automaton for $\cR_{2^k-1}$ and
this completes the proof of Proposition~\ref{p: savitch}.

\subsection{Proof of Lemma~\ref{l: R_n_k not in PA_k}}
\label{ss: proof}

The proof of Lemma~\ref{l: R_n_k not in PA_k} is rather long and technical.
This subsection and the next are devoted to it.

Recall that for each $i\in\{0,1,2,\ldots\}$,
we define $n_i = 2^{i+1} - 2$.
An equivalent recursive definition is $n_0 = 0$, and
$n_{i} = 2 n_{i-1} + 2$, when $i \geq 1$.

By Theorem~\ref{t: equivalence},
it is sufficient to consider only one-way deterministic PA $\A$.
Let $\A = \langle Q,q_0, \mu, F\rangle$ be a strong $k$-PA.
By adding some extra states,
we can normalise the behaviour of each pebble as follows.
For each $i\in \{1,\ldots,k\}$, pebble~$i$ behaves as follows.
\begin{itemize}
\item
After pebble~$i$ moves right and $i> 1$, 
then pebble~$(i-1)$ is immediately placed
(in position 0 reading the left end-marker $\triangleleft$).
\item
If $i < k$, pebble~$i$ is lifted only when
it reaches the right-end marker $\triangleright$
of the input.
\item
Immediately after pebble~$i$ is lifted, pebble~$(i+1)$ moves right.
\end{itemize}
We also assume that in the automaton $\A$ 
only pebble~$k$ can enter a final state and
it may do so only after it reads the right-end marker $\triangleright$
of the input.

We define the following integers:
$\beta_0 = 1$, $\beta_1 = |Q|$, and for $i \geq 2$,\footnote{$!$ denotes factorial.}
\[
\beta_{i}  =  |Q|! \times \beta_{i-1}!
\]
For the rest of this subsection and the next,
we fix the integers $k$ and $m$,
where $k$ is the number of pebbles of $\A$ and $m=\beta_{k+1}$.

We define the following graph $G_{n_k,m} = (V_{n_k,m},E_{n_k,m})$.
The set $V_{k,m}$ consists of the following vertices.
\begin{itemize}
\item
$a_0,a_1,\ldots,a_{n_k}$;
\item
$b_0,b_1,\ldots,b_{n_k-1}$;
\item
$c_{1,i},\ldots,c_{n_k-1,i}$, for each $i=1,\ldots,m-1$; and
\item
$d_{1,i},\ldots,d_{n_k-1,i}$, for each $i=1,\ldots,m-1$,
\end{itemize}
where $a_0,\ldots,a_{n_k}, b_0,\ldots,b_{n_k-1},
c_{1,1},\ldots,c_{n_k-1,m-1},d_{1,1},\ldots,d_{n_k-1,m-1}$
are all different.
The set $E_{k,m}$ consists of the following edges.
\begin{itemize}
\item
$(a_0,a_1),(a_1,a_2),\ldots,(a_{n_k-1},a_{n_k})$;
\item
$(b_0,b_1),(b_1,b_2),\ldots,(b_{n_k-2},b_{n_k}-1)$; 
\item
$(c_{1,i},c_{2,i}),(c_{2,i},c_{3,i}),\ldots,(c_{n_k-2,i},c_{n_k-1,i})$, 
for each $i=1,\ldots,m-1$; and
\item
$(d_{1,i},d_{2,i}),(d_{2,i},d_{3,i}),\ldots,(d_{n_k-2,i},d_{n_k-1,i})$, 
for each $i=1,\ldots,m-1$.
\end{itemize}
Figure~1 below illustrates the graph $G_{n_k,m}$.

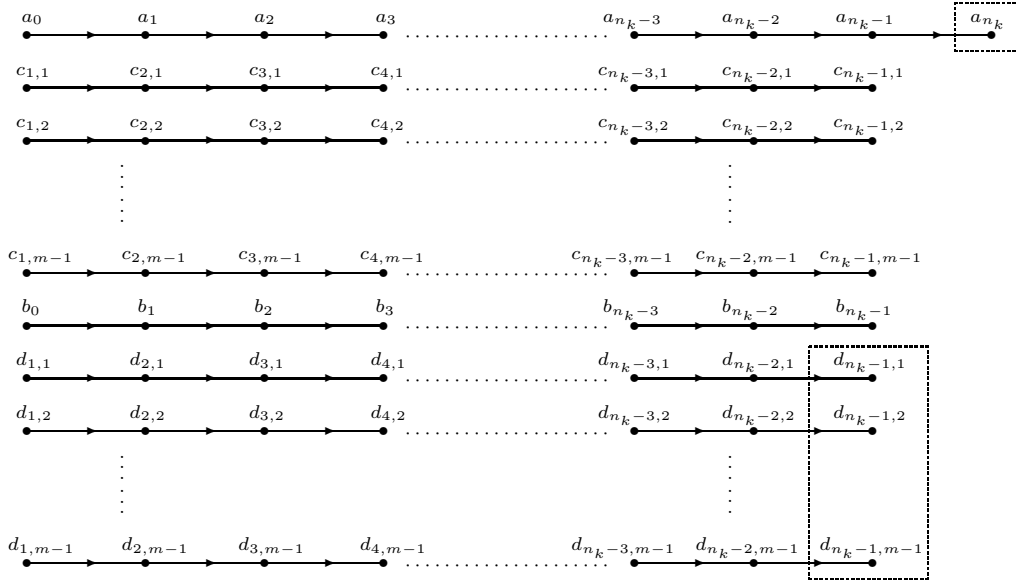
\begin{figure}[!ht]
\label{real graph}
\begin{picture}(350,250)(-12,-125)

{\scriptsize

\multiput(-10,100)(45,0){3}{\vector(1,0){27}}
\multiput(-10,100)(45,0){3}{\line(1,0){45}}
\multiput(-10,100)(45,0){4}{\circle*{3}}
\put(-12,105){$a_0$}
\put(32,105){$a_1$}
\put(76,105){$a_2$}
\put(122,105){$a_3$}
\put(133,97.5){$\cdots\cdots\cdots\cdots\cdots\cdots\cdots$}
\multiput(220,100)(45,0){3}{\vector(1,0){27}}
\multiput(220,100)(45,0){3}{\line(1,0){45}}
\multiput(220,100)(45,0){4}{\circle*{3}}
\put(208,105){$a_{{n_k}-3}$}
\put(253,105){$a_{{n_k}-2}$}
\put(296,105){$a_{{n_k}-1}$}
\put(347,105){$a_{n_k}$}

\multiput(-10,80)(45,0){3}{\vector(1,0){27}}
\multiput(-10,80)(45,0){3}{\line(1,0){45}}
\multiput(-10,80)(45,0){4}{\circle*{3}}
\put(-14,85){$c_{1,1}$}
\put(29,85){$c_{2,1}$}
\put(74,85){$c_{3,1}$}
\put(120,85){$c_{4,1}$}
\put(133,77.5){$\cdots\cdots\cdots\cdots\cdots\cdots\cdots$}
\multiput(220,80)(45,0){2}{\vector(1,0){27}}
\multiput(220,80)(45,0){2}{\line(1,0){45}}
\multiput(220,80)(45,0){3}{\circle*{3}}
\put(206,85){$c_{{n_k}-3,1}$}
\put(253,85){$c_{{n_k}-2,1}$}
\put(295,85){$c_{{n_k}-1,1}$}

\multiput(-10,60)(45,0){3}{\vector(1,0){27}}
\multiput(-10,60)(45,0){3}{\line(1,0){45}}
\multiput(-10,60)(45,0){4}{\circle*{3}}
\put(-14,65){$c_{1,2}$}
\put(29,65){$c_{2,2}$}
\put(74,65){$c_{3,2}$}
\put(120,65){$c_{4,2}$}
\put(133,57.5){$\cdots\cdots\cdots\cdots\cdots\cdots\cdots$}
\multiput(220,60)(45,0){2}{\vector(1,0){27}}
\multiput(220,60)(45,0){2}{\line(1,0){45}}
\multiput(220,60)(45,0){3}{\circle*{3}}
\put(206,65){$c_{{n_k}-3,2}$}
\put(253,65){$c_{{n_k}-2,2}$}
\put(295,65){$c_{{n_k}-1,2}$}

\multiput(25,42)(0,-13){2}{$\vdots$}
\multiput(255,42)(0,-13){2}{$\vdots$}

\multiput(-10,10)(45,0){3}{\vector(1,0){27}}
\multiput(-10,10)(45,0){3}{\line(1,0){45}}
\multiput(-10,10)(45,0){4}{\circle*{3}}
\put(-17,15){$c_{1,m-1}$}
\put(26,15){$c_{2,m-1}$}
\put(70,15){$c_{3,m-1}$}
\put(116,15){$c_{4,m-1}$}
\put(133,7.5){$\cdots\cdots\cdots\cdots\cdots\cdots\cdots$}
\multiput(220,10)(45,0){2}{\vector(1,0){27}}
\multiput(220,10)(45,0){2}{\line(1,0){45}}
\multiput(220,10)(45,0){3}{\circle*{3}}
\put(196,15){$c_{{n_k}-3,m-1}$}
\put(243,15){$c_{{n_k}-2,m-1}$}
\put(290,15){$c_{{n_k}-1,m-1}$}

\multiput(-10,-10)(45,0){3}{\vector(1,0){27}}
\multiput(-10,-10)(45,0){3}{\line(1,0){45}}
\multiput(-10,-10)(45,0){4}{\circle*{3}}
\put(-12,-5){$b_0$}
\put(32,-5){$b_1$}
\put(76,-5){$b_2$}
\put(122,-5){$b_3$}
\put(133,-12.5){$\cdots\cdots\cdots\cdots\cdots\cdots\cdots$}
\multiput(220,-10)(45,0){2}{\vector(1,0){27}}
\multiput(220,-10)(45,0){2}{\line(1,0){45}}
\multiput(220,-10)(45,0){3}{\circle*{3}}
\put(208,-5){$b_{{n_k}-3}$}
\put(253,-5){$b_{{n_k}-2}$}
\put(296,-5){$b_{{n_k}-1}$}

\multiput(-10,-30)(45,0){3}{\vector(1,0){27}}
\multiput(-10,-30)(45,0){3}{\line(1,0){45}}
\multiput(-10,-30)(45,0){4}{\circle*{3}}
\put(-14,-25){$d_{1,1}$}
\put(29,-25){$d_{2,1}$}
\put(74,-25){$d_{3,1}$}
\put(120,-25){$d_{4,1}$}
\put(133,-32.5){$\cdots\cdots\cdots\cdots\cdots\cdots\cdots$}
\multiput(220,-30)(45,0){2}{\vector(1,0){27}}
\multiput(220,-30)(45,0){2}{\line(1,0){45}}
\multiput(220,-30)(45,0){3}{\circle*{3}}
\put(206,-25){$d_{{n_k}-3,1}$}
\put(253,-25){$d_{{n_k}-2,1}$}
\put(295,-25){$d_{{n_k}-1,1}$}

\multiput(-10,-50)(45,0){3}{\vector(1,0){27}}
\multiput(-10,-50)(45,0){3}{\line(1,0){45}}
\multiput(-10,-50)(45,0){4}{\circle*{3}}
\put(-14,-45){$d_{1,2}$}
\put(29,-45){$d_{2,2}$}
\put(74,-45){$d_{3,2}$}
\put(120,-45){$d_{4,2}$}
\put(133,-52.5){$\cdots\cdots\cdots\cdots\cdots\cdots\cdots$}
\multiput(220,-50)(45,0){2}{\vector(1,0){27}}
\multiput(220,-50)(45,0){2}{\line(1,0){45}}
\multiput(220,-50)(45,0){3}{\circle*{3}}
\put(206,-45){$d_{{n_k}-3,2}$}
\put(253,-45){$d_{{n_k}-2,2}$}
\put(295,-45){$d_{{n_k}-1,2}$}

\multiput(25,-68)(0,-13){2}{$\vdots$}
\multiput(255,-68)(0,-13){2}{$\vdots$}

\multiput(-10,-100)(45,0){3}{\vector(1,0){27}}
\multiput(-10,-100)(45,0){3}{\line(1,0){45}}
\multiput(-10,-100)(45,0){4}{\circle*{3}}
\put(-17,-95){$d_{1,m-1}$}
\put(26,-95){$d_{2,m-1}$}
\put(70,-95){$d_{3,m-1}$}
\put(116,-95){$d_{4,m-1}$}
\put(133,-102.5){$\cdots\cdots\cdots\cdots\cdots\cdots\cdots$}
\multiput(220,-100)(45,0){2}{\vector(1,0){27}}
\multiput(220,-100)(45,0){2}{\line(1,0){45}}
\multiput(220,-100)(45,0){3}{\circle*{3}}
\put(196,-95){$d_{{n_k}-3,m-1}$}
\put(243,-95){$d_{{n_k}-2,m-1}$}
\put(290,-95){$d_{{n_k}-1,m-1}$}

\put(341,94){\dashbox(25,18){}}
\put(286,-106){\dashbox(45,88){}}

}
\end{picture}
\caption{The full graph is the graph $G_{n_k,m}$.
The graph depicted by $\overline{w}(n_k,m)$ is also the above graph but
without the nodes inside the dashed box and the edges adjacent to them.}
\end{figure}

Now consider the following word $w(n_k,m)$:
\begin{equation}
\label{eq: w_{k,m}}
w(n_k,m)  = 
a_0a_1 C_1 b_0b_1  D_1  
\cdots\cdots 
a_{n_k-2}a_{n_k-1}  C_{n_k-1}  b_{n_k-2}b_{n_k-1}  D_{n_k-1} a_{n_k-1}a_{n_k}
\end{equation}
where for each $i=0,1,\ldots,n_k-2$,
\begin{itemize}
\item
$C_i = c_{i,1}c_{i+1,1}\ \cdots \ c_{i,m-1}c_{i+1,m-1} $;
\item
$D_i = d_{i,1}d_{i+1,1}\ \cdots \ d_{i,m-1}d_{i+1,m-1} $.
\end{itemize}
This word $w(n_k,m)$ induces the graph $G_{n_k,m}$,
that is, $G_{w(n_k,m)}= G_{n_k,m}$
and $s_{w(n_k,m)}=a_0$ and $t_{w(n_k,m)}=a_{n_k}$.

Now let 
\begin{equation}
\label{eq: bar w_{k,m}}
\overline{w}(n_k,m) =
a_0a_1  C_1  b_0b_1  D_1  
\cdots\cdots
a_{n_k-2}a_{n_k-1}  C_{n_k-1}  b_{n_k-2}b_{n_k-1}.
\end{equation}
That is, the word $\overline{w}(n_k,m)$
is obtained by deleting the suffix $D_{n_k-1} a_{n_k-1}a_{n_k}$
from $w(n_k,m)$.

The graph $G_{\overline{w}(n_k,m)}$ is also illustrated in the graph in Figure~1,
the graph $G_{\overline{w}(n_k,m)}$ is without
the nodes inside the dashed box and the edges adjacent to them.

and note that $s_{\overline{w}(n_k,m)}=a_0$ and $t_{\overline{w}(n_k,m)}=b_{n_k-1}$.
Obviously, $w(n_k,m) \in \cR_{n_k}$, 
while $\overline{w}(n_k,m) \notin \cR$.

To prove Lemma~\ref{l: R_n_k not in PA_k},
we are going to prove the following proposition.
\begin{proposition}
\label{p: accept or reject both}
The automaton $\A$ either accepts both $w(n_k,m)$ and $\overline{w}(n_k,m)$,
or rejects both $w(n_k,m)$ and $\overline{w}(n_k,m)$.
\end{proposition}

The proof is rather complicated.
It consist of five claims and their interdependence is illustrated below.\footnote{We 
are going to prove Claims~\ref{cl: proof subclaim B one} and~\ref{cl: proof subclaim B two}
by induction simultaneously. This will be made precise in Subsection~\ref{ss: proof subclaim}.}

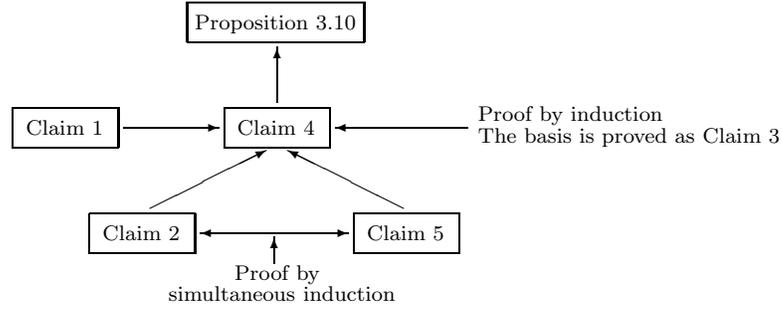
\begin{picture}(350,160)(-200,-110)

{\footnotesize

\put(-83,10){\line(1,0){67}}
\put(-83,10){\line(0,1){15}}
\put(-83,25){\line(1,0){67}}
\put(-16,10){\line(0,1){15}}
\put(-80,15){Proposition~\ref{p: accept or reject both}}

\put(-49,-13){\vector(0,1){21}}
\put(-69,-30){\line(1,0){40}}
\put(-69,-30){\line(0,1){15}}
\put(-69,-15){\line(1,0){40}}
\put(-29,-30){\line(0,1){15}}
\put(-64,-25){Claim~\ref{cl: main claim}}

\put(23,-22.5){\vector(-1,0){50}}
\put(27,-20){Proof by induction}
\put(27,-28){The basis is proved as Claim~\ref{cl: basis main claim}}

\put(-107,-22.5){\vector(1,0){36}}
\put(-149,-30){\line(1,0){40}}
\put(-149,-30){\line(0,1){15}}
\put(-149,-15){\line(1,0){40}}
\put(-109,-30){\line(0,1){15}}
\put(-144,-25){Claim~\ref{cl: compatibility preserved}}

\put(-97,-53){\vector(2,1){44}}
\put(-120,-70){\line(1,0){40}}
\put(-120,-70){\line(0,1){15}}
\put(-120,-55){\line(1,0){40}}
\put(-80,-70){\line(0,1){15}}
\put(-115,-65){Claim~\ref{cl: proof subclaim B one}}

\put(-78,-62.5){\vector(1,0){56}}
\put(-78,-62.5){\vector(-1,0){0}}
\put(-50,-74){\vector(0,1){10}}
\put(-65,-80){Proof by}
\put(-90,-88){simultaneous induction}

\put(-1,-53){\vector(-2,1){44}}
\put(-20,-70){\line(1,0){40}}
\put(-20,-70){\line(0,1){15}}
\put(-20,-55){\line(1,0){40}}
\put(20,-70){\line(0,1){15}}
\put(-15,-65){Claim~\ref{cl: proof subclaim B two}}

}
\end{picture}

In the proof we will need quite a number of notions which, 
for the sake of readability, are listed below one-by-one
before we define them properly.
\begin{itemize}
\item
The notions of $K(l)$ and $L(l)$.
\item
The notion of successor of a pebble assignment.
\item
The notion of compatibility between two pebble assignments.
\end{itemize}

\paragraph*{\bf The notions of $K(l)$ and $L(l)$}
For $l \in \{0,1,\ldots,n_k-1\}$,
we define the integers $K(l)$ and $L(l)$
which are illustrated as follows.

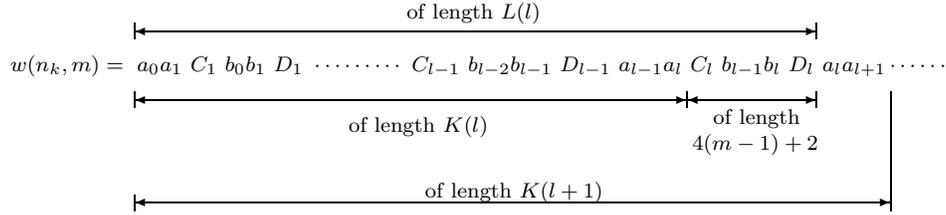
\begin{picture}(350,110)(-200,-70)

{\footnotesize

\put(-160,18){\line(0,-1){6}}
\put(-160,15){\vector(-1,0){0}}
\put(-57,20){of length $L(l)$}
\put(98,15){\vector(1,0){0}}
\put(98,18){\line(0,-1){6}}
\put(-160,15){\line(1,0){258}}

\put(-207,0){$w(n_k,m)=$}
\put(-159,0){$a_0a_1 \ C_1 \ b_0b_1 \ D_1
\ \cdots\cdots\cdots \
C_{l-1} \ b_{l-2}b_{l-1} \ D_{l-1} \ a_{l-1}a_{l} \ C_{l} \ b_{l-1}b_{l} \ D_{l}
\ a_la_{l+1} \cdots\cdots$}

\put(-160,-8){\line(0,-1){6}}

\put(-160,-11){\vector(-1,0){0}}
\put(-79,-23){of length $K(l)$}
\put(49,-11){\vector(1,0){0}}

\put(49,-8){\line(0,-1){6}}

\put(49,-11){\vector(-1,0){0}}
\put(59,-20){of length}
\put(51,-30){$4(m-1)+2$}
\put(98,-11){\vector(1,0){0}}

\put(98,-8){\line(0,-1){6}}
\put(-160,-11){\line(1,0){258}}

\put(-50,-48){of length $K(l+1)$}
\put(-160,-50){\line(1,0){286}}
\put(-160,-50){\vector(-1,0){0}}
\put(-160,-53){\line(0,1){6}}
\put(126,-53){\line(0,1){45}}
\put(126,-50){\vector(1,0){0}}

}
\end{picture}

Formally, for $l \in \{0,1,\ldots,n_k\}$, 
\begin{eqnarray*}
K(l) 
& = &
\left\{
\begin{array}{ll}
 0, & \quad\mbox{if} \ l=0
 \\
 4m(l-1) + 2, & \quad\mbox{if} \ l \geq 1
\end{array}
\right.
\end{eqnarray*}
and for $l \in \{0,1,\ldots,n_k\}$,
\begin{eqnarray*}
L(l) 
& = &
\left\{
\begin{array}{ll}
 K(l+1)-2, & \quad\mbox{if} \ l \leq n_k-1
 \\
 K(n_k), & \quad\mbox{otherwise}.
\end{array}
\right.
\end{eqnarray*}
In particular, $K(n_k)$ is precisely the length of the word $w(n_k,m)$
and $L(0)=0$.

\paragraph*{\bf The notion of successor of a pebble assignment}

Let $\theta$ be an assignment of pebbles~$i,i+1,\ldots,k$ of $\A$ on a word $w$.
That is, $\theta$ is a function from $\{i,i+1,\ldots,k\}$ to $\{0,1,\ldots,|w|+1\}$.
(Recall that positions~$0$ and $|w|+1$ contain 
the left- and right-end markers $\triangleleft$ and $\triangleright$, respectively.)
If $0\leq \theta(i)\leq |w|$,
we define $Succ_i(\theta) = \theta'$, where
for each $j \in \{i,i+1,\ldots,k\}$,
$$
\theta'(j) =
\left\{
\begin{array}{ll}
\theta (j) & \quad\mbox{if} \ j \geq i+1
\\
\theta(i)+1 & \quad\mbox{if} \ j = i
\end{array}
\right.
$$

\paragraph*{\bf The notion of compatibility between two configurations}
Let $i \geq 1$ and $[i,q,\theta]$ and $[i,\overline{q},\overline{\theta}]$
be configurations of $\A$ on $w(n_k,m)$ and $\overline{w}(n_k,m)$, respectively,
when pebble~$i$ is the head pebble.
For an integer $l \in \{0,1,\ldots,n_k\}$,
we say that the configurations $[i,q,\theta]$ and $[i,\overline{q},\overline{\theta}]$
are {\em compatible} with respect to $l$,
if
\begin{itemize}
\item
$q=\overline{q}$;
\end{itemize}
and for each $j \in \{i,\ldots,k\}$,
\begin{itemize}
\item
either $\theta(j) \leq K(l)$ or $\theta(j) \geq L(l+n_{i})$;
\item
either $\overline{\theta}(j)\leq K(l)$ or $\overline{\theta}\geq L(l+n_{i})-2m$;
\item
if $\theta(j) \leq K(l)$, then $\overline{\theta}(j)\leq K(l)$
and $\theta(j) = \overline{\theta}(j)$;
\item
if $\overline{\theta}(j) \leq K(l)$, then $\theta(j)\leq K(l)$
and $\theta(j) = \overline{\theta}(j)$;
\item
if $\theta(j) \geq L(l+n_{i})$, then $\overline{\theta}(j)\geq L(l+n_{i})-2m$
and $\theta(j) = \overline{\theta}(j)+2m$;
\item
if $\overline{\theta}(j) \geq L(l+n_{i})-2m$, then $\theta(j)\geq L(l+n_{i})$
and $\theta(j) = \overline{\theta}(j)+2m$.
\end{itemize}
Below we give an illustration of the compatibility of 
two configurations of an $8$-PA on $w(n_8,m)$ and $\overline{w}(n_8,m)$,
respectively, with respect to $l$. The index $\ell$ is $l+n_5$.

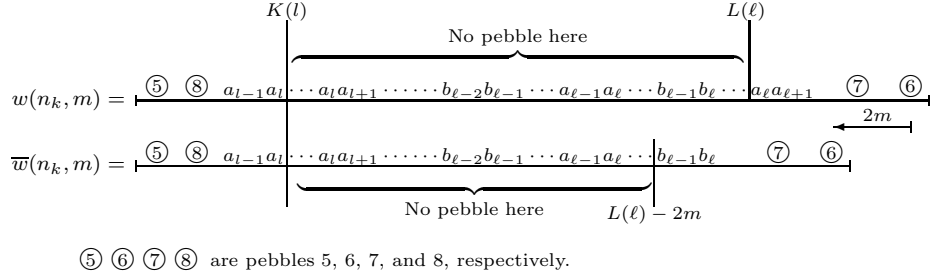
\begin{picture}(350,120)(-200,-60)

{\footnotesize

\put(-104,47){\scriptsize $K(l)$}
\put(-96,45){\line(0,-1){70}}
%\put(-153,30){\vector(1,0){30}}
%\put(-153,30){\vector(-1,0){0}}

\put(-94,28){$\overbrace{\hspace{6.0 cm}}^{\textrm{\scriptsize No pebble here}}$}

\put(70,47){\scriptsize $L(\ell)$}
\put(79,45){\line(0,-1){30}}
%\put(-153,30){\vector(1,0){30}}
%\put(-153,30){\vector(-1,0){0}}

\put(-200,13.5){$w(n_k,m)=$}
\put(-153,17){\line(0,-1){4}}
\put(147,17){\line(0,-1){4}}
\put(-153,15){\line(1,0){300}}

\put(-145,20.5){\circle{8}}
\put(-147,18){\scriptsize 5}
\put(-130,20.5){\circle{8}}
\put(-132,18){\scriptsize 8}

\put(-120,17){\scriptsize $a_{l-1}a_l \cdots a_la_{l+1} 
\cdots\cdots b_{\ell-2}b_{\ell-1} \cdots a_{\ell-1}a_{\ell}\cdots 
b_{\ell-1}b_{\ell}\cdots a_{\ell}a_{\ell+1} $}

\put(120,20.5){\circle{8}}
\put(118,18){\scriptsize 7}
\put(140,20.5){\circle{8}}
\put(138,18){\scriptsize 6}

\put(-92,-15){$\underbrace{\hspace{4.7 cm}}_{\textrm{\scriptsize No pebble here}}$}
\put(43,-25){\line(0,1){25}}
\put(25,-30){\scriptsize $L(\ell)-2m$}

\put(-200,-11.5){$\overline{w}(n_k,m)=$}
\put(-153,-8){\line(0,-1){4}}
\put(117,-8){\line(0,-1){4}}
\put(-153,-10){\line(1,0){270}}

\put(-145,-4.5){\circle{8}}
\put(-147,-7){\scriptsize 5}
\put(-130,-4.5){\circle{8}}
\put(-132,-7){\scriptsize 8}

\put(-120,-8){\scriptsize $a_{l-1}a_l \cdots a_la_{l+1} 
\cdots\cdots b_{\ell-2}b_{\ell-1} \cdots a_{\ell-1}a_{\ell}\cdots 
b_{\ell-1}b_{\ell} $}

%\put(72,7){\scriptsize $2m$}
%\put(90,5){\vector(-1,0){30}}
%\put(90,3){\line(0,1){4}}
\put(90,-4.5){\circle{8}}
\put(88,-7){\scriptsize 7}

\put(122,7){\scriptsize $2m$}
\put(140,5){\vector(-1,0){29}}
\put(140,3){\line(0,1){4}}
\put(110,-4.5){\circle{8}}
\put(108,-7){\scriptsize 6}

\put(-170,-45){\circle{8}}
\put(-172,-47.5){\scriptsize 5}
\put(-158,-45){\circle{8}}
\put(-160,-47.5){\scriptsize 6}
\put(-146,-45){\circle{8}}
\put(-148,-47.5){\scriptsize 7}
\put(-134,-45){\circle{8}}
\put(-136,-47.5){\scriptsize 8}
\put(-125,-47.5){\scriptsize are pebbles~5, 6, 7, and 8, respectively.}

}
\end{picture}

\begin{claim}
\label{cl: compatibility preserved}
Suppose that $[i,q,\theta]$ and $[i,q,\overline{\theta}]$
are configurations of $\A$ on $w(n_k,m)$ and $\overline{w}(n_k,m)$, respectively.
If $[i,q,\theta]$ and $[i,q,\overline{\theta}]$ are compatible
with respect to some $l \in \{0,\ldots,n_k\}$,
then
\begin{enumerate}
\item
for all $h \in \{0,\ldots,K(l+n_{i-1}+2)\}$ and for all $p \in Q$,
the configuration $[i-1,p,\theta\cup\{(i-1,h)\}]$ (on $w(n_k,m)$)
and the configuration $[i-1,p,\overline{\theta}\cup\{(i-1,h)\}]$ (on $\overline{w}(n_k,m)$)
are compatible with respect to $l + n_{i-1}+2$;
\item
for all $h \in \{L(l+n_{i-1}),\ldots,K(n_k)\}$ and for all $p \in Q$,
the configuration $[i-1,p,\theta\cup\{(i-1,h)\}]$ (on $w(n_k,m)$)
and the configuration $[i-1,p,\overline{\theta}\cup\{(i-1,h-2m)\}]$ (on $\overline{w}(n_k,m)$)
are compatible with respect to $l$.
\end{enumerate}
\end{claim}
\begin{proof}
It follows from the fact that $n_{i} = 2n_{i-1} + 2$. 
We prove it by picture here.
For case~(1), the proof is as follows. Let $l' = l + n_{i}$.

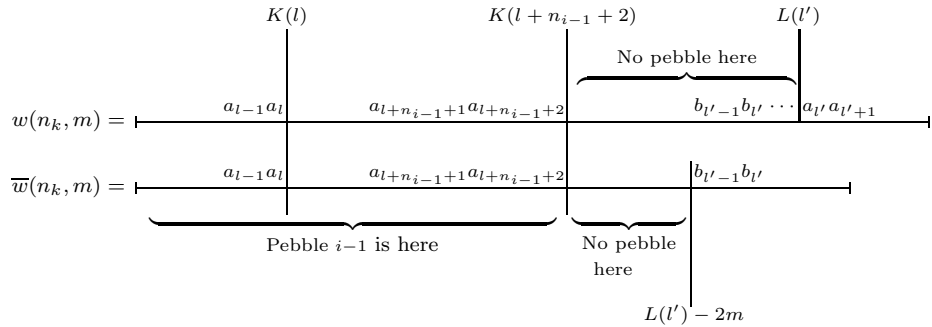
\begin{picture}(350,145)(-200,-75)

{\footnotesize

\put(-200,13.5){$w(n_k,m)=$}
\put(-153,17){\line(0,-1){4}}
\put(147,17){\line(0,-1){4}}
\put(-153,15){\line(1,0){300}}

\put(-120,19){\scriptsize $a_{l-1}a_l$}

\put(-104,53){\scriptsize $K(l)$}
\put(-96,50){\line(0,-1){70}}

\put(14,28){$\overbrace{\hspace{2.9 cm}}^{\textrm{\scriptsize No pebble here}}$}

\put(-20,53){\scriptsize $K(l+n_{i-1}+2)$}
\put(10,50){\line(0,-1){70}}

\put(-65,19){\scriptsize $a_{l+n_{i-1}+1}a_{l+n_{i-1}+2}$}

\put(58,19){\scriptsize $b_{l'-1}b_{l'}\cdots$}

\put(89,53){\scriptsize $L(l')$}
\put(98,50){\line(0,-1){35}}

\put(99,19){\scriptsize $a_{l'}a_{l'+1}$}

\put(-148,-20){$\underbrace{\hspace{5.4 cm}}_{\textrm{\scriptsize Pebble}\ i-1 \ \textrm{is here}}$}

\put(12,-20){$\underbrace{\hspace{1.5 cm}}_{\textrm{\scriptsize No pebble}}$}
\put(20,-43){\scriptsize here}
\put(57,-55){\line(0,1){55}}
\put(39,-60){\scriptsize $L(l')-2m$}

\put(-200,-11.5){$\overline{w}(n_k,m)=$}
\put(-153,-8){\line(0,-1){4}}
\put(117,-8){\line(0,-1){4}}
\put(-153,-10){\line(1,0){270}}

\put(-120,-6){\scriptsize $a_{l-1}a_l$}
\put(-65,-6){\scriptsize $a_{l+n_{i-1}+1}a_{l+n_{i-1}+2}$}
\put(58,-6){\scriptsize $b_{l'-1}b_{l'}$}

}
\end{picture}

There is no pebble on the positions between $K(l+n_{i-1}+2)$ and $L(')$
in the word $w(n_k,m)$ as well as on the positions
between $K(l+n_{i-1}+2)$ and $L(')-2m$ in the word $\overline{w}(n_k,m)$
due to the assumption that $[i,q,\theta]$ and $[i,q,\overline{\theta}]$
are compatible with respect to $l$.
Since $l' - (l+n_{i-1}+2) = n_{i-1}$,
the configuration $[i-1,p,\theta\cup\{(i-1,h)\}]$ (on $w(n_k,m)$)
and the configuration $[i-1,p,\overline{\theta}\cup\{(i-1,h)\}]$ (on $\overline{w}(n_k,m)$)
are compatible with respect to $l + n_{i-1}+2$,
for all $h \in \{0,\ldots,K(l+n_{i-1}+2)\}$ and for all $p \in Q$.

For case~(2), the proof is as follows.
We let $l'' = l+ n_{i-1}$.

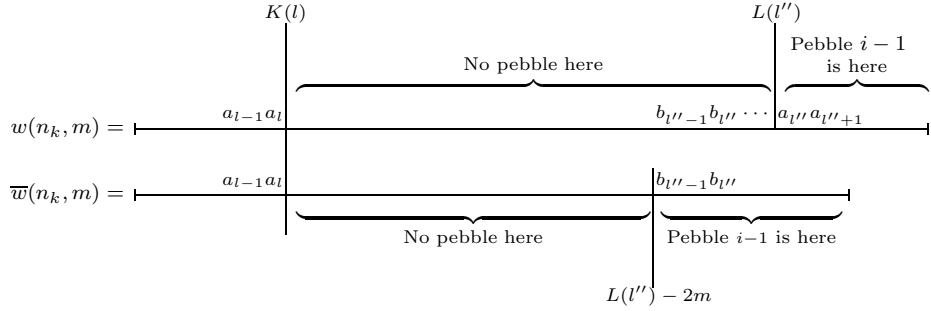
\begin{picture}(350,145)(-200,-70)

{\footnotesize

\put(-104,57){\scriptsize $K(l)$}
\put(-96,55){\line(0,-1){80}}
%\put(-153,30){\vector(1,0){30}}
%\put(-153,30){\vector(-1,0){0}}

\put(-92,28){$\overbrace{\hspace{6.3 cm}}^{\textrm{\scriptsize No pebble here}}$}
\put(95,45){$\textrm{\scriptsize Pebble} \ i-1$}
\put(93,28){$\overbrace{\hspace{1.9 cm}}^{\textrm{\scriptsize is here}}$}

\put(80,57){\scriptsize $L(l'')$}
\put(89,55){\line(0,-1){40}}
%\put(-153,30){\vector(1,0){30}}
%\put(-153,30){\vector(-1,0){0}}

\put(-200,13.5){$w(n_k,m)=$}
\put(-153,17){\line(0,-1){4}}
\put(147,17){\line(0,-1){4}}
\put(-153,15){\line(1,0){300}}

\put(-120,19){\scriptsize $a_{l-1}a_l$}
\put(44,19){\scriptsize $b_{l''-1}b_{l''}\cdots$}
\put(90,19){\scriptsize $a_{l''}a_{l''+1}$}

\put(-92,-15){$\underbrace{\hspace{4.7 cm}}_{\textrm{\scriptsize No pebble here}}$}
\put(43,-45){\line(0,1){45}}
\put(25,-50){\scriptsize $L(l'')-2m$}

\put(-200,-11.5){$\overline{w}(n_k,m)=$}
\put(-153,-8){\line(0,-1){4}}
\put(117,-8){\line(0,-1){4}}
\put(-153,-10){\line(1,0){270}}

\put(-120,-6){\scriptsize $a_{l-1}a_l$}
\put(44,-6){\scriptsize $b_{l''-1}b_{l''}$}

\put(46,-15){$\underbrace{\hspace{2.4 cm}}_{\textrm{\scriptsize Pebble}\ i-1 \ \textrm{\scriptsize is here}}$}

}
\end{picture}

There is no pebble on the positions between $K(l)$ and $L('')$
in the word $w(n_k,m)$ as well as on the positions
between $K(l)$ and $L('')-2m$ in the word $\overline{w}(n_k,m)$
due to the assumption that $[i,q,\theta]$ and $[i,q,\overline{\theta}]$
are compatible with respect to $l$.
Hence, case~(2) follows immediately.
This completes the proof of Claim~\ref{cl: compatibility preserved}.
\end{proof}

\begin{remark}
\label{rem: applied transition}
Let $[i,q,\theta]$ and $[i,q,\overline{\theta}]$
be configurations of $\A$ on $w(n_k,m)$ and $\overline{w}(n_k,m)$, respectively
and assume that they are compatible with respect to an integer $l$.
Let $j,j' \in \{i,i+1,\ldots,k\}$ and
let 
\begin{itemize}
\item
$x$ and $y$ denote the symbols seen by pebbles~$j$ and $j'$, respectively,
on $w(n_k,m)$ according to the configuration $\theta$, and 
\item
$\overline{x}$ and $\overline{y}$ denote the symbols seen by pebbles~$j$ and $j'$, respectively,
on $\overline{w}(n_k,m)$ according to the configuration $\overline{\theta}$.
\end{itemize}
Then $x= y$ if and only if $\overline{x}=\overline{y}$.

The reason is as follows.
Since $[i,q,\theta]$ and $[i,q,\overline{\theta}]$
are compatible with respect to $l$,
we have the following four cases.
\begin{itemize}
\item[(a)]
$\theta(j)\leq K(l)$ and $\theta(j') \leq K(l)$.
\\
In this case, $\overline{\theta}(j)=\theta(j)$ and $\overline{\theta}(j')=\theta(j')$
and we immediately have $x=y$ if and only if $\overline{x}=\overline{y}$.
\item[(b)]
$\theta(j)\leq K(l)$ and $\theta(j') \geq L(l+n_i)$.
\\
In this case, $\overline{\theta}(j)=\theta(j)$ and $\overline{\theta}(j')=\theta(j')-2m$.
Now in $w(n_k,m)$ and $\overline{w}(n_k,m)$
each symbol appears at most twice and
they are of distance $4m-2$ apart.
Since $L(l+n_i) - K(l) > 4m-2$, we have $x\neq y$.
Similarly, $L(l+n_i)-2m - K(l) > 4m-2$, hence $\overline{x}=\overline{y}$.
\item[(c)]
$\theta(j)\geq L(l+n_i)$ and $\theta(j') \leq K(l)$.
\\
The proof is similar to case~(b) above.
\item[(d)]
$\theta(j)\geq L(l+n_i)$ and $\theta(j') \geq L(l+n_i)$.
\\
In this case, $\overline{\theta}(j)=\theta(j)-2m$ and $\overline{\theta}(j')=\theta(j')-2m$
and we immediately have $x=y$ if and only if $\overline{x}=\overline{y}$.
\end{itemize}
Now this immediately implies that
for every transition $\alpha \to \beta$ of the automaton $\A$,
it applies to $[i,q,\theta]$ if and only if 
it applies to $[i,q,\overline{\theta}]$.
\end{remark}

The following claim is important.
However, due to the complexity of its proof, 
we postpone it until Subsection~\ref{ss: proof subclaim}.

\begin{claim}
\label{cl: proof subclaim B one}
For each $i \in \{1,\ldots,k\}$, and for every run of $\A$ on $w(n_k,m)$:
\begin{equation}
\label{eq: cl: proof subclaim B one}
[i,p_0,\theta_0] \ \vdash^{\ast}_{\sA,w(n_k,m)} 
[i,p_1,\theta_1] \ \vdash^{\ast}_{\sA,w(n_k,m)}
\ \cdots\cdots   \ \vdash^{\ast}_{\sA,w(n_k,m)} \ [i,p_{N+1},\theta_{N+1}]
\end{equation}
where  
\begin{itemize}
\item
$N= K(n_k)=$ length of $w(n_k,m)$;
\item
$\theta_0(i) = 0$;
\item
$\theta_{N+1}(i) = N+1$;
\item
$\theta_{h+1}=Succ_i(\theta_h)$, for each $h \in \{0,\ldots,N\}$ --
that is, for each $j \in \{i+1,\ldots,k\}$, 
$\theta_0(j) = \cdots = \theta_{N+1}(j)$ and
$\theta_h(i) = h$, for each $h \in \{0,\ldots,N+1\}$;
\end{itemize}
if $l$ is an integer  such that
\begin{enumerate}
\item
if $i=k$, then $l=0$; and
\item
if $i \neq k$, then $l$ is an integer such that
for each $j \in \{i+1,\ldots,k\}$,
either $\theta(j) \leq K(l)$, or $\theta(j) \geq L(l+n_i)+1$,
\end{enumerate}
then there exist two positive integers $\nu_0$ and $\nu$
such that
\begin{itemize}
\item
$\nu=\pi \beta_{i-1}!$, where $1\leq\pi\leq |Q|$;
\item
$K(l+n_{i-1}+1)+1 \leq \nu_0 \leq K(l+n_{i-1}+1) + \beta_i$;
\item
for each $h$ where $\nu_0 \leq h \leq K(l+n_{i-1}+2)-\nu$,
we have $p_{h}=p_{h+\nu}$.
\end{itemize}
In particular, since $\beta_{i+1}=|Q|!\times \beta_i!$ and $m=\beta_{k+1}$, 
we have $\nu$ divides $\beta_{i+1}$, 
and thus $\nu$ also divides $m$.
Therefore, $p_{K(l+n_{i-1}+2)-2-2m} = p_{K(l+n_{i-1}+2)-2}$.
\end{claim}

Below we give an illustration of the intuitive meaning 
of the indexes $l,\nu_0,\nu$ in Claim~\ref{cl: proof subclaim B one} for $i \neq k$.
Let $l$ be the integer assumed in the hypothesis of Claim~\ref{cl: proof subclaim B one}.
(For simplicity, we do not put the indexes on the $a$'s.)

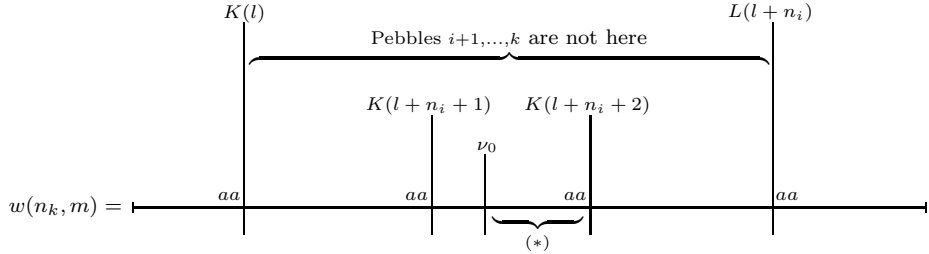
\begin{picture}(350,125)(-200,-40)

{\footnotesize

\put(-200,-16.5){$w(n_k,m)=$}
\put(-153,-13){\line(0,-1){4}}
\put(147,-13){\line(0,-1){4}}
\put(-153,-15){\line(1,0){300}}

\put(-119,57){\scriptsize $K(l)$}
\put(-111,55){\line(0,-1){80}}
\put(-121,-12){\scriptsize $aa$}

\put(-109,38){$\overbrace{\hspace{6.9 cm}}^{\textrm{\scriptsize Pebbles} 
\ i+1,\ldots,k \ \textrm{are not here}}$}

\put(72,57){\scriptsize $L(l+n_i)$}
\put(89,55){\line(0,-1){80}}
\put(90,-12){\scriptsize $aa$}

\put(-65,22){\scriptsize $K(l+n_i+1)$}
\put(-40,20){\line(0,-1){45}}
\put(-50,-12){\scriptsize $aa$}

\put(-5,22){\scriptsize $K(l+n_i+2)$}
\put(20,20){\line(0,-1){45}}
\put(10,-12){\scriptsize $aa$}

\put(-23,7){\scriptsize $\nu_0$}
\put(-20,5){\line(0,-1){30}}

\put(-17,-17){$\underbrace{\hspace{1.2 cm}}_{(*)}$}

}
\end{picture}
The meaning of Claim~\ref{cl: proof subclaim B one}
is that in region~($\ast$) pebble~$i$
enters the same state every $\nu$ steps.

\begin{claim}
\label{cl: basis main claim}
Let 
$$
[1,p_0,\theta_0] \ \vdash_{\sA,w(n_k,m)} 
\ \cdots\cdots \ 
[1,p_N,\theta_N] \ \vdash_{\sA,w(n_k,m)} \ 
[1,p_{N+1},\theta_{N+1}]
$$
be a run of $\A$ on $w(n_k,m)$,
where $N$ is the length of $w(n_k,m)$ and $\theta_0(1) = 0$, and 
$\theta_{j+1} = Succ_1(\theta_{j})$, for each $j \in \{0,\ldots,N\}$;
and let
$$
[1,r_0,\overline{\theta}_0] \ \vdash_{\sA,\overline{w}(n_k,m)} 
\ \cdots\cdots \ 
[1,r_M,\overline{\theta}_M] \ \vdash_{\sA,\overline{w}(n_k,m)} \ 
[1,r_{M+1},\overline{\theta}_{M+1}]
$$
be a run of $\A$ on $\overline{w}(n_k,m)$,
where $M$ is the length of $\overline{w}(n_k,m)$ and $\overline{\theta}_0(1) = 0$, and 
$\overline{\theta}_{j+1} = Succ_1(\overline{\theta}_{j})$, for each $j \in \{0,\ldots,M\}$.

If $[1,p_0,\theta_0]$ and $[1,r_0,\overline{\theta}_0]$
are compatible with respect to an $l \in \{0,\ldots,n_k-n_1\}$,
then $p_{N+1}=r_{M+1}$.
\end{claim}
\begin{proof}
Consider the run 
$$
[1,p_0,\theta_0] \ \vdash_{\sA,w(n_k,m)} 
\ \cdots\cdots \ 
[1,p_N,\theta_N] \ \vdash_{\sA,w(n_k,m)} \ 
[1,p_{N+1},\theta_{N+1}],
$$
where $\theta_0(1) = 0$, and 
$\theta_{j+1} = Succ_1(\theta_{j})$, for each $j \in \{0,\ldots,N\}$;
and the run
$$
[1,r_0,\overline{\theta}_0] \ \vdash_{\sA,\overline{w}(n_k,m)} 
\ \cdots\cdots \ 
[1,r_M,\overline{\theta}_M] \ \vdash_{\sA,\overline{w}(n_k,m)} \ 
[1,r_{M+1},\overline{\theta}_{M+1}],
$$
where $\overline{\theta}_0(1) = 0$, and 
$\overline{\theta}_{j+1} = Succ_1(\overline{\theta}_{j})$, for each $j \in \{0,\ldots,M\}$.

Suppose that $[1,p_0,\theta_0]$ and $[1,r_0,\theta_0]$ are compatible
with respect to an integer $l$.
This means that $p_0 = r_0$.
We are going to show that $p_{N+1} = r_{M+1}$ in three stages. 
(In the following let $l' = l+2$.)
\begin{description}
\item[Stage 1]
$p_{K(l')}=r_{K(l')}$.

To prove this,
we show that $p_h = r_h$, for each $h \in \{0,\ldots,K(l')\}$.
The proof is by induction on $h$.
The proof for the base case, $h=0$, 
follows from compatibility of $[1,p_0,\theta_0]$ and $[1,r_0,\overline{\theta}_0]$.

For the induction step, suppose that $p_h = r_h$.
By Remark~\ref{rem: applied transition},
a transition $\alpha\to \beta$ applies to $[1,p_h,\theta_h]$
if and only if it applies to $[1,r_h,\overline{\theta}_h]$.
Hence, $p_{h+1}=r_{h+1}$.

\item[Stage 2] 
$p_{K(l')-2}=p_{K(l')-2m-2}=r_{K(l')-2m-2}$.

In Stage~1, we already show that $p_{K(l')-2m-2}=r_{K(l')-2m-2}$.
That $p_{K(l')-2}=p_{K(l')-2m-2}$ follows from Claim~\ref{cl: proof subclaim B one}.

% This follows from Claim~\ref{cl: proof subclaim B one},
% which for the sake of readability,
% is proved in Subsection~\ref{ss: proof subclaim} below.
% The main idea is that since the integer $m$ is big enough,
% there exists an integer $\nu$ such that
% on every $\nu$ steps, 
% pebble~$1$ will enter into the same state.
% The integer $m$ is defined so that 
% it is divisible by $\nu$,
% which implies $p_{K(l')-2}=p_{K(l')-2m-2}$.

\item[Stage 3] 
$p_{N+1}=r_{M+1}$. 
 
We are going to prove that $p_h = r_{h-2m}$,
for each $h \in \{K(l')-2,\ldots,N+1\}$.

The proof is by induction on $h$.
The proof for the base case, $h=K(l')-2$, 
is already shown in Step~2.

For the induction step, suppose that $p_h = r_{h-2m}$.
By Remark~\ref{rem: applied transition},
a transition $\alpha\to \beta$ applies to $[1,p_h,\theta_h]$
if and only if it applies to $[1,r_{h-2m},\overline{\theta}_{h-2m}]$.
Thus, $p_{h+1}=r_{h+1}$.
\end{description}
This completes the proof of Claim~\ref{cl: basis main claim}.
\end{proof}

The following claim is the generalisation of Claim~\ref{cl: basis main claim}
which implies Proposition~\ref{p: accept or reject both}.
\begin{claim}
\label{cl: main claim}
For each $i \in \{1,\ldots,k\}$, the following holds.
Let 
$$
[i,p_0,\theta_0] \ \vdash^{\ast}_{\sA,w(n_k,m)} 
\ \cdots\cdots \ 
[i,p_N,\theta_N] \ \vdash^{\ast}_{\sA,w(n_k,m)} \ 
[i,p_{N+1},\theta_{N+1}]
$$
be a run of $\A$ on $w(n_k,m)$,
where $N$ is the length of $w(n_k,m)$ and $\theta_0(i) = 0$, and 
$\theta_{j+1} = Succ_i(\theta_{j})$, for each $j \in \{0,\ldots,N\}$;
and let
$$
[i,r_0,\overline{\theta}_0] \ \vdash^{\ast}_{\sA,\overline{w}(n_k,m)} 
\ \cdots\cdots \ 
[i,r_M,\overline{\theta}_M] \ \vdash^{\ast}_{\sA,\overline{w}(n_k,m)} \ 
[i,r_{M+1},\overline{\theta}_{M+1}]
$$
be a run of $\A$ on $\overline{w}(n_k,m)$,
where $M$ is the length of $\overline{w}(n_k,m)$ and $\overline{\theta}_0(i) = 0$, and 
$\overline{\theta}_{j+1} = Succ_i(\overline{\theta}_{j})$, for each $j \in \{0,\ldots,M\}$.

If $[i,p_0,\theta_0]$ and $[i,r_0,\overline{\theta}_0]$
are compatible with respect to an $l \in \{0,\ldots,n_k-n_i\}$,
then $p_{N+1}=r_{M+1}$.
\end{claim}
\begin{proof}
The proof is by induction on $i$.
The basis is $i=1$, which we have already proved in Claim~\ref{cl: basis main claim}.

For the induction hypothesis, we assume that Claim~\ref{cl: main claim} holds
for the case of $i-1$.
We are going to show that it holds for the case of $i$.
The line of reasoning is almost the same as Claim~\ref{cl: basis main claim}.
For completeness, we present it here.

Consider the following run  
$$
[i,p_0,\theta_0] \ \vdash^{\ast}_{\sA,w(n_k,m)} 
\ \cdots\cdots \ 
[i,p_N,\theta_N] \ \vdash^{\ast}_{\sA,w(n_k,m)} \ [i,p_{N+1},\theta_{N+1}]
$$
and
$$
[i,r_0,\overline{\theta}_0] \
\vdash^{\ast}_{\sA,\overline{w}(n_k,m)} \ \cdots\cdots \ \vdash^{\ast}_{\sA,\overline{w}(n_k,m)} \
[i,r_{M+1},\overline{\theta}_{M+1}].
$$
By the assumption that 
$[i,p_0,\theta_0]$ and $[i,r_0,\theta_0]$ are compatible,
we have $p_0 = r_0$.
We are going to prove that $p_{N+1} = r_{M+1}$ in three stages.
Let $l' = l+ n_{i-1}+2$.
\begin{description}
\item[Stage 1] 
$p_{K(l')-2}=r_{K(l')-2}$.

To prove this subclaim,
we show that $p_h = r_h$, for each $h \in \{0,\ldots,K(l')\}$.
The proof is by induction on $h$.
The proof for the base case $p_0=r_0$ follows from
the fact that $[i,p_0,\theta_0]$ and $[i,r_0,\theta_0]$ are compatible.

For the induction step, suppose that $p_h = r_h$.
By the normalisation of the automaton $\A$,
the run is of the form: 
$$
[i,p_h,\theta_h] \ \vdash_{\sA,w(n_k,m)} 
[i-1,p_0',\theta_0'] \ \vdash^{\ast}_{\sA,w(n_k,m)} 
\ \cdots\cdots \ 
\vdash^{\ast}_{\sA,w(n_k,m)} \ [i-1,p_{N+1}',\theta_{N+1}']
$$
and
$$
[i,r_h,\overline{\theta}_h] \
\vdash_{\sA,\overline{w}(n_k,m)}
[i-1,r_0',\overline{\theta}_0'] \
\vdash^{\ast}_{\sA,\overline{w}(n_k,m)} \ \cdots\cdots \ \vdash^{\ast}_{\sA,\overline{w}(n_k,m)} \
[i,r_{M+1}',\overline{\theta}_{M+1}'],
$$
where $\theta'_h(i-1)=h$ for each $h \in \{1,\ldots,N+1\}$
and $\overline{\theta}'_h(i-1)=h$ for each $h \in \{1,\ldots,M+1\}$.

By determinism of $\A$, we have $p_0'=r_0'$.
Then, by Claim~\ref{cl: compatibility preserved},
since $0\leq h \leq K(l')$, we have
$[i-1,p_0',\theta_0']$ and $[i-1,r_0',\overline{\theta}_0']$
compatible with respect to $l'$.
By the induction hypothesis of Claim~\ref{cl: main claim},
we have $p_{N+1}'=r_{M+1}'$.
Then, by determinism of $\A$, we have $p_{h+1} = r_{h+1}$.

\item[Stage 2] 
$p_{K(l')-2}=p_{K(l')-2m-2}=r_{K(l')-2m-2}$.

In Stage~1 we already have $p_{K(l')-2m-2}=r_{K(l')-2m-2}$.
Claim~\ref{cl: proof subclaim B one} implies that $p_{K(l')-2}=p_{K(l')-2m-2}$.

\item[Stage 3] $p_{N+1}=r_{M+1}$. 
 
By Subclaim~B, we have $p_{K(l')-2}=r_{K(l')-2m-2}$.
We are going to prove that $p_h = r_{h-2m}$,
for each $h \in \{K(l')-2,\ldots,N+1\}$.

The proof is by induction on $h$.
The proof for the base case, $h=K(l')$, follows from Subclaim~B.

For the induction step, suppose that $p_h = r_{h-2m}$.
By the normalisation of the automaton $\A$,
we assume that the run is of the form: 
$$
[i,p_h,\theta_h] \ \vdash_{\sA,w(n_k,m)} 
[i-1,p_0',\theta_0'] \ \vdash^{\ast}_{\sA,w(n_k,m)} 
\ \cdots\cdots \ 
\vdash^{\ast}_{\sA,w(n_k,m)} \ [i-1,p_{N+1}',\theta_{N+1}']
$$
and
$$
[i,r_{h-2m},\overline{\theta}_{h-2m}] \
\vdash_{\sA,\overline{w}(n_k,m)}
[i-1,r_0',\overline{\theta}_0'] \
\vdash^{\ast}_{\sA,\overline{w}(n_k,m)} \ \cdots\cdots \ \vdash^{\ast}_{\sA,\overline{w}(n_k,m)} \
[i,r_{M+1}',\overline{\theta}_{M+1}'],
$$
where $\theta'_h(i-1)=h$ for each $h \in \{1,\ldots,N+1\}$
and $\overline{\theta}'_h(i-1)=h$ for each $h \in \{1,\ldots,M+1\}$.

That we have 
$[i,p_h,\theta_h]  \vdash_{\sA,w(n_k,m)}  [i-1,p_0',\theta_0']$
and 
$[i,r_{h-2m},\overline{\theta}_h] \vdash_{\sA,\overline{w}(n_k,m)} [i-1,r_0',\overline{\theta}_0']$
is due to the normalisation of the automaton $\A$
described in the beginning of Subsection~\ref{ss: proof}.

By determinism of $\A$, we have $p_0'=r_0'$.
Then, by Claim~\ref{cl: compatibility preserved},
since $h \geq K(l')$, we have
$[i-1,p_0',\theta_0']$ and $[i-1,r_0',\overline{\theta}_0']$
compatible with respect to $l$.
By the induction hypothesis of Claim~\ref{cl: main claim},
we have $p_{N+1}'=r_{M+1}'$.
Then, by determinism of $\A$, we have $p_{h+1} = r_{h+1-2m}$.
\end{description}
This completes the proof of Claim~\ref{cl: main claim}.
\end{proof}

\begin{proof} (of Proposition~\ref{p: accept or reject both})
We simply apply Claim~\ref{cl: main claim},
in which $i=k$, and both $p_0,r_0$ are the initial state $q_0$ of $\A$.
Note that the initial configurations of $\A$ on $w(n_k,m)$ and $\overline{w}(n_k,m)$
are the same, thus, they are compatible.
\end{proof}

\subsection{Proof of Claim~\ref{cl: proof subclaim B one}}
\label{ss: proof subclaim}

In this subsection we are going to prove
Claim~\ref{cl: proof subclaim B one}.
The proof is also rather long and technical.
We need the following definition.
\begin{definition}
\label{d: pebble-ordering}
In the following, let $i\in\{1,\ldots,k\}$.
\begin{enumerate}
\item
An assignment $\theta:\{i,\ldots,k\}\mapsto \{0,1,\ldots,K(n_k)+1\}$
of pebbles~$i,i+1,\ldots,k$ on $w(n_k,m)$ is called a pebble-$i$ assignment.
\item
For two pebble-$i$ assignments $\theta_1$ and $\theta_2$,
we say that they have the same pebble ordering,
if for each $j,j' \in \{i,i+1,\ldots,k\}$,
$\theta_1(j) \leq \theta_1(j')$ if and only if
$\theta_2(j) \leq \theta_2(j')$.
\end{enumerate}
\end{definition}

In this subsection we are going to prove 
Claim~\ref{cl: proof subclaim B one} together with Claim~\ref{cl: proof subclaim B two} below.
In fact, we are going to prove {\em both} claims simultaneously.
(We will give the structure of the proofs later on.)

\begin{claim}
\label{cl: proof subclaim B two}
Let $[i,q,\theta_1]$ and $[i,q,\theta_2]$
be configurations of $\A$ on $w(n_k,m)$ such that
\begin{enumerate}
\item
$\theta_1$ and $\theta_2$ have the same pebble ordering;
\item
for each $j \in \{i,\ldots,k\}$, $\theta_1(j) \leq \theta_2(j)$;
\item
there exist integers $l_1,l_2,l_3,l_4$ and $\pi$ such that
$l_1 \leq l_2 \leq l_3 \leq l_4$ and
$1\leq\pi < \frac{m}{\beta_{i-1}!}$ and
for each $j \in \{i,\ldots,k\}$,
\begin{enumerate}
\item
if $\theta_1(j) \leq K(l_1)$ or $\theta_1(j) \geq L(l_4)+1$,
then $\theta_1(j) = \theta_2(j)$;
\item
if $\theta_2(j) \leq K(l_1)$ or $\theta_2(j) \geq L(l_4)+1$,
then $\theta_1(j) = \theta_2(j)$;
\item
$l_2 - l_1 \geq n_{i-1}+1$;
\item
$l_4 - l_3 \geq n_{i-1}+1$;
\item
$\sfImage(\theta_1) \cap (\{K(l_1)+1,\ldots,K(l_2)\}\cup \{L(l_3)+1,\ldots,L(l_4)\}) = \emptyset$;
\item
$\sfImage(\theta_2) \cap (\{K(l_1)+1,\ldots,K(l_2)\}\cup \{L(l_3)+1,\ldots,L(l_4)\}) = \emptyset$;
\item
if $\theta_1(j) \in \{K(l_2)+1,\ldots,L(l_3)\}$,
then $\theta_2(j) \in \{K(l_2)+1,\ldots,L(l_3)\}$
and $\theta_2(j) - \theta_1(j) = \pi \beta_{i-1}!$;
\item
if $\theta_2(j) \in \{K(l_2)+1,\ldots,L(l_3)\}$,
then $\theta_1(j) \in \{K(l_2)+1,\ldots,L(l_3)\}$
and $\theta_2(j) - \theta_1(j) = \pi \beta_{i-1}!$.
\end{enumerate}
\end{enumerate}
If $[i,q,\theta_1] \vdash^{\ast} [i,p,Succ_i(\theta_1)]$
and $[i,q,\theta_2] \vdash^{\ast} [i,r,Succ_i(\theta_2)]$,
then $p=r$.
\end{claim}

Below we give an intuitive meaning of Claim~\ref{cl: proof subclaim B two}.
Consider the following illustration,
where $\theta_1$ and $\theta_2$ are configurations on $w(n_k,m)$ with the same pebble ordering.

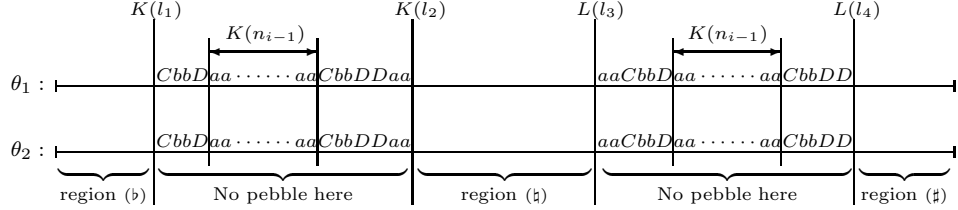
\begin{picture}(350,120)(-200,-60)

{\footnotesize

\put(-146,40){\line(0,-1){70}}
\put(-155,42){\scriptsize $K(l_1)$}

\put(-125,33){\line(0,-1){48}}
\put(-119,33){\scriptsize $K(n_{i-1})$}
\put(-125,28){\vector(1,0){41}}
\put(-125,28){\vector(-1,0){0}}
\put(-84,33){\line(0,-1){48}}

\put(-48,40){\line(0,-1){70}}
\put(-55,42){\scriptsize $K(l_2)$}

\put(21,40){\line(0,-1){70}}
\put(14,42){\scriptsize $L(l_3)$}

\put(50.5,33){\line(0,-1){48}}
\put(56,33){\scriptsize $K(n_{i-1})$}
\put(50.5,28){\vector(1,0){41}}
\put(50.5,28){\vector(-1,0){0}}
\put(91.5,33){\line(0,-1){48}}

\put(119,40){\line(0,-1){70}}
\put(111,42){\scriptsize $L(l_4)$}

%\put(-94,28){$\overbrace{\hspace{6.0 cm}}^{\textrm{\scriptsize No pebble here}}$}

\put(-200,13.5){$\theta_1:$}
\put(-183,17){\line(0,-1){4}}
\put(157,17){\line(0,-1){4}}
\put(-183,15){\line(1,0){340}}

%\put(-175,20.5){\circle{8}}
%\put(-177,18){\scriptsize 5}
%\put(-160,20.5){\circle{8}}
%\put(-162,18){\scriptsize 8}

\put(-145,17){\scriptsize $C  bb  D  aa \cdots\cdots aa  C  bb  DD aa$}

\put(22,17){\scriptsize $aa C  bb  D  aa \cdots\cdots aa  C  bb  DD$}

%\put(145,20.5){\circle{8}}
%\put(143,18){\scriptsize 7}
%\put(0,20.5){\circle{8}}
%\put(-2,18){\scriptsize 6}

\put(-182,-15){$\underbrace{\hspace{1.2 cm}}_{\textrm{\scriptsize region}~(\flat)}$}

\put(-144,-15){$\underbrace{\hspace{3.3 cm}}_{\textrm{\scriptsize No pebble here}}$}

\put(-46,-15){$\underbrace{\hspace{2.3 cm}}_{\textrm{\scriptsize region}~(\natural)}$}

\put(24,-15){$\underbrace{\hspace{3.3 cm}}_{\textrm{\scriptsize No pebble here}}$}

\put(121,-15){$\underbrace{\hspace{1.2 cm}}_{\textrm{\scriptsize region}~(\sharp)}$}

\put(-200,-11.5){$\theta_2:$}
\put(-183,-8){\line(0,-1){4}}
\put(157,-8){\line(0,-1){4}}
\put(-183,-10){\line(1,0){340}}

%\put(-175,-4.5){\circle{8}}
%\put(-177,-7){\scriptsize 5}
%\put(-160,-4.5){\circle{8}}
%\put(-162,-7){\scriptsize 8}

\put(-145,-8){\scriptsize $C  bb  D  aa \cdots\cdots aa  C  bb  DD aa$}

\put(22,-8){\scriptsize $aa C  bb  D  aa \cdots\cdots aa  C  bb  DD$}

%\put(145,-4.5){\circle{8}}
%\put(143,-7){\scriptsize 7}

%\put(-13,7){\scriptsize $\pi\beta_4$}
%\put(0,3){\vector(-1,0){15}}
%\put(0,4){\line(0,-1){2}}
%\put(-15,-4.5){\circle{8}}
%\put(-17,-7){\scriptsize 6}

}
\end{picture}

The meanings of $l_1,l_2,l_3,l_4$ and $\pi$ are such that
for each $j \in \{i,\ldots,k\}$,
\begin{itemize}
\item
if pebble~$j$ are found in region~$(\flat)$ on both configurations~$\theta_1$ and $\theta_2$,
then $\theta_1(j)=\theta_2(j)$;
\item
if pebble~$j$ are found in region~$(\natural)$ on both configurations~$\theta_1$ and $\theta_2$,
then $\theta_2(j)-\theta_1(j) = \pi \beta_{i-1}!$;
\item
if pebble~$j$ are found in region~$(\sharp)$ on both configurations~$\theta_1$ and $\theta_2$,
then $\theta_2(j)=\theta_1(j)$.
\end{itemize}
On both configurations $\theta_1$ and $\theta_2$
no pebbles are found in the region between $K(l_1)+1$ and $K(l_2)$
as well as in between $L(l_3)+1$ and $L(l_4)$.
Claim~\ref{cl: proof subclaim B two} states that
both configurations $[i,q,\theta_1]$ and $[i,q,\theta_2]$ are essentially the ``same.''
In the sense that if $[i,q,\theta_1] \vdash^{\ast} [i,p,Succ_i(\theta_1)]$
and $[i,q,\theta_2] \vdash^{\ast} [i,r,Succ_i(\theta_2)]$,
then $p=r$.

The proofs of both Claims~\ref{cl: proof subclaim B one} and~\ref{cl: proof subclaim B two} 
use a rather involved inductive argument.
In fact, we are going to prove {\em both} claims simultaneously by induction.
The induction step on the proof of each claim uses
the induction hypothesis of both claims.
The overall structure of the proofs of both
Claims~\ref{cl: proof subclaim B one}
and~\ref{cl: proof subclaim B two} is as follows.
\begin{enumerate}
\item
We prove the base case $i=1$ of Claim~\ref{cl: proof subclaim B one}.
\item
We prove the base case $i=1$ of Claim~\ref{cl: proof subclaim B two}.
\item
For the induction hypothesis, 
we assume that {\em both} Claims~\ref{cl: proof subclaim B one}
and~\ref{cl: proof subclaim B two} hold for the case $i$.
\item
For the induction step,
we prove Claim~\ref{cl: proof subclaim B one}
for the case $i+1$.

This step uses the hypothesis that both
Claims~\ref{cl: proof subclaim B one} and~\ref{cl: proof subclaim B two} hold
for case $i$. 
\item
For the other induction step, 
we prove Claim~\ref{cl: proof subclaim B two}
for the case $i+1$.

As in Step~4, this step uses the hypothesis that both
Claims~\ref{cl: proof subclaim B one} and~\ref{cl: proof subclaim B two} hold
for case $i$. 
\end{enumerate}

\paragraph*{\bf Proof of the base case $i=1$ for Claim~\ref{cl: proof subclaim B one}}

Let $l$ be an integer such that for each $j \in \{2,\ldots,k\}$,
either $\theta(j) \leq K(l)$ or $\theta(j)\geq L(l+2)$,
where the number $2$ comes from $n_1 = 2$.

The symbols in $C_{l+1}b_{l}b_{l+1}D_{l+1}$
are different from all the symbols seen by pebbles~$2,\ldots,k$.
We are going to show that when reading $C_{l+1}b_{l}b_{l+1}D_{l+1}$,
pebble~1 enters into a loop of states.
See the illustration below.

\noindent
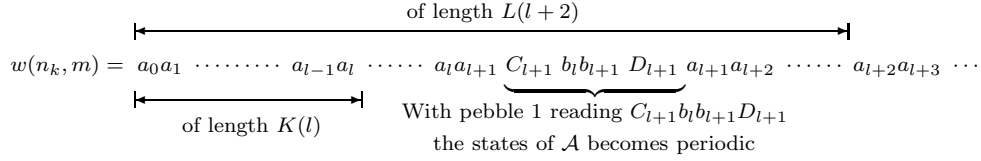
\begin{picture}(350,110)(-200,-60)

{\footnotesize

\put(-153,18){\line(0,-1){6}}
\put(-153,15){\vector(-1,0){0}}
\put(-50,20){of length $L(l+2)$}
\put(117,15){\vector(1,0){0}}
\put(117,18){\line(0,-1){6}}
\put(-153,15){\line(1,0){270}}

\put(-200,0){$w(n_k,m)=$}
\put(-152,0){$a_0a_1 \ \cdots\cdots\cdots \ a_{l-1}a_{l} 
\ \cdots\cdots \ 
a_{l}a_{l+1} \ C_{l+1} \ b_{l}b_{l+1} \ D_{l+1} \
a_{l+1}a_{l+2} \ \cdots \cdots \  a_{l+2}a_{l+3}
\ \cdots$}

\put(-153,-8){\line(0,-1){6}}

\put(-153,-11){\vector(-1,0){0}}
\put(-135,-23){of length $K(l)$}
\put(-67,-11){\vector(1,0){0}}

\put(-67,-8){\line(0,-1){6}}
\put(-153,-11){\line(1,0){86}}

\put(-52,-4){$\underbrace{\hspace{2.4 cm}}_{\textrm{With pebble~1 reading 
$C_{l+1} b_{l}b_{l+1} D_{l+1}$}}$}
\put(-40,-30){the states of $\A$ becomes periodic}

}
\end{picture}

On reading the segment $C_{l+1}b_{l}b_{l+1}D_{l+1}$,
the transitions used are of the form $(1,\emptyset,\emptyset,s)\to(s',\ttRight)$.
Due to the determinism of the automaton $\A$,
there exist integers $\nu_0$ and $\nu$ such that
$\nu_0,\nu \leq |Q|$ and
for each $h$ where $\nu_0 \leq h \leq K(l+n_{i-1}+2)-\nu$,
we have $p_{h}=p_{h+\nu}$.
In particular, since $\beta_2 = |Q|! \times \beta_1!$,
we have $\nu$ divides $\beta_2$.
Furthermore, $\beta_2$ also divides $m=\beta_{k+1}$,
thus, $\nu$ divides $m$, therefore,
$p_{K(l+n_{i-1}+2)-2-2m} = p_{K(l+n_{i-1}+2)-2}$.

\paragraph*{\bf Proof of the base case $i=1$ for Claim~\ref{cl: proof subclaim B two}}

Suppose $[1,q,\theta_1]$ and $[1,q,\theta_2]$
are configurations of $\A$ on $w(n_k,m)$
and $l_1,l_2,l_3,l_4,\pi$ are integers
such that the conditions $(1)$, $(2)$, $(3.a)$-$(3.h)$ above hold.
Moreover, suppose also that
$$
[1,q,\theta_1]\vdash [1,p,Succ_1(\theta_1)]
\quad \mbox{and} \quad
[1,q,\theta_2]\vdash [1,r,Succ_1(\theta_2)].
$$
We are going to show that $p=r$.

By conditions $(3.e)$ and $(3.f)$,
there can only be three cases: $\theta_1(1)\leq K(l_1)$,
$K(l_2)+1 \leq \theta_1(1) \leq L(l_3)$, and
$\theta_1(1) \geq L(l_4)+1$.

\begin{description}
\item[Case 1]
$\theta_1(1)\leq K(l_1)$.

By condition $(3.a)$,
we have $\theta_1(1) = \theta_2(1)$.
By conditions~$(1)$, $(2)$, $(3.a)$ and $(3.b)$,
for any $j \in \{2,\ldots,k\}$,
we have 
\begin{equation}
\label{eq: same P}
\theta_1(j)=\theta_1(1) \quad \mbox{if and only if} \quad \theta_2(j)=\theta_2(1).
\end{equation}
By condition~$(3.c)$, $l_2-l_1 \geq 1$.
Moreover, no symbol in $C_{l_2} b_{l_{2}-1}b_{l_{2}} D_{l_2}\cdots a_{n_k-1}a_{n_k}$
appears in $a_0a_1 \cdots a_{l_1-1}a_{l_1}$,
and by conditions~$(3.e)$ and~$(3.f)$,
no pebbles are placed on $C_{l_1}\cdots a_{l_2-1}a_{l_2}$.
Therefore,
for any $j \in \{2,\ldots,k\}$,
\begin{equation}
\label{eq: same V}
\begin{array}{c}
\textrm{pebbles~$j$ and $1$ read the same symbol in the configuration $[1,q,\theta_1]$}
\\
\mbox{if and only if}
\\
\textrm{pebbles~$j$ and $1$ read the same symbol in the configuration $[1,q,\theta_2]$}
\end{array}
\end{equation}
Thus, by Equalities~\ref{eq: same P} and~\ref{eq: same V},
the same transition applies to both $[1,q,\theta_1]$ and $[1,q,\theta_2]$.
Since $\A$ is deterministic, we have $p=r$.

\item[Case 2]
$K(l_2)+1 \leq \theta_1(1) \leq L(l_3)$.

That is, $\theta_2(1) = \theta_1(1) + \pi\beta_{i-1}!$, where $1\leq \pi\beta_{i-1}! < m$.
By the same conditions $(3.g)$ and $(3.h)$,
for any $j \in \{2,\ldots,k\}$,
\begin{equation}
\label{eq: same P case 2}
\theta_1(j)=\theta_1(1) \quad \mbox{if and only if} \quad \theta_2(j)=\theta_2(1).
\end{equation}
By condition~$(3.c)$, $l_2-l_1 \geq 1$.
Moreover, any symbol in $C_{l_2} b_{l_{2}-1}b_{l_{2}} D_{l_2}\cdots a_{n_k-1}a_{n_k}$
does not appear in $a_0a_1 \cdots a_{l_1-1}a_{l_1}$.
Therefore,
for any $j \in \{2,\ldots,k\}$,
if pebbles~$j$ and $1$ read the same symbol in the configuration $[1,q,\theta_1]$,
then $K(l_2)+1 \leq \theta_1(j) \leq L(l_3)$;
and similarly, 
if pebbles~$j$ and $1$ read the same symbol in the configuration $[1,q,\theta_2]$,
then $K(l_2)+1 \leq \theta_2(j) \leq L(l_3)$.
By conditions~$(3.g)$ and~$(3.h)$, $\theta_2(j) = \theta_1(j) + \pi\beta_{i-1}!$.
Due to the definition of $w(n_k,m)$,
we have
\begin{equation}
\label{eq: same V case 2}
\begin{array}{c}
\textrm{pebbles~$j$ and $1$ read the same symbol in the configuration $[1,q,\theta_1]$}
\\
\mbox{if and only if}
\\
\textrm{pebbles~$j$ and $1$ read the same symbol in the configuration $[1,q,\theta_2]$}
\end{array}
\end{equation}
Thus, by Equalities~\ref{eq: same P case 2} and~\ref{eq: same V case 2},
the same transition applies to both $[1,q,\theta_1]$ and $[1,q,\theta_2]$.
Since $\A$ is deterministic, we have $p=r$.

\item[Case 3]
$\theta_1(1) \geq L(l_4)+1$.

The proof is similar to the one for Case~1 above, thus, omitted.
\end{description}
This completes the proof of the base case $i=1$ 
for Claim~\ref{cl: proof subclaim B two}.

\paragraph*{\bf The induction hypothesis} 
Both Claims~\ref{cl: proof subclaim B one} and~\ref{cl: proof subclaim B two} hold
for case $i$.

\paragraph*{\bf The induction step for Claim~\ref{cl: proof subclaim B one}}
We are going to show that Claim~\ref{cl: proof subclaim B one} holds for the case $i+1$.

Suppose we have the following run:
$$
[i+1,p_0,\theta_0] \ \vdash^{\ast}_{\sA,w(n_k,m)} 
[i+1,p_1,\theta_1] \ \vdash^{\ast}_{\sA,w(n_k,m)}
\ \cdots\cdots   \ \vdash^{\ast}_{\sA,w(n_k,m)} \ 
[i+1,p_{N+1},\theta_{N+1}]
$$
Let $l$ be the integer as stated in Claim~\ref{cl: proof subclaim B one}.
Since $m > |Q|\beta_i!$,
there exists a pair $(\eta,\eta')$ of indexes such that
\begin{itemize}
\item
$K(l+n_{i}+1)+1 \leq \eta < \eta' \leq K(l+n_{i}+2) -2$;
\item
$\eta'-\eta = \pi \beta_i!$, where $1\leq \pi \leq |Q|$;
\item
$p_{\eta}=p_{\eta'}$.
\end{itemize}
We pick such pair $(\eta,\eta')$ in which $\eta$ is the smallest.
We claim that $\nu_0=\eta$ and $\nu = \eta'-\eta$
are the desired two integers in Claim~\ref{cl: proof subclaim B one}.

We are going to show that
for each $h \in \{\nu_0,\ldots,K(l+n_{i}+2)-2-\nu\}$,
\begin{equation}
\label{eq: inductive proof subclaim B one}
\textrm{if $p_{h} = p_{h+\nu}$, then $p_{h+1} = p_{h+\nu+1}$}.
\end{equation}
Since by definition of $\nu_0$ and $\nu$, we already have
$p_{\nu_0} = p_{\nu_0+\nu}$, this immediately implies
that for each $h \in \{\nu_0,\ldots,K(l+n_{i}+2)-2-\nu\}$,
$p_{h}=p_{h+\nu}$.

To prove Equality~\ref{eq: inductive proof subclaim B one},
suppose $p_{h} = p_{h+\nu}$.
Consider the following run:
\begin{itemize}
\item
$[i+1,p_{h},\theta_h] \vdash [i,s_0,\theta_h\cup\{(i,0)\}]$;
\item
$[i,s_0,\theta_h\cup\{(i,0)\}] \vdash^{\ast}
\cdots \vdash^{\ast}
[i,s_{N+1},\theta_h\cup\{(i,N+1)\}]$
\item
$[i,s_{N+1},\theta_h\cup\{(i,N+1)\}] \vdash [i+1,s',\theta_h]
\vdash [i+1, p_{h+1}, \theta_{h+1}]$.
\end{itemize}
and the following run:
\begin{itemize}
\item
$[i+1,p_{h+\nu},\theta_{h+\nu}] \vdash [i,t_0,\theta_{h+\nu}\cup\{(i,0)\}]$;
\item
$[i,t_0,\theta_{h+\nu}\cup\{(i,0)\}] \vdash^{\ast}
\cdots \vdash^{\ast}
[i,t_{N+1},\theta_{h+\nu}\cup\{(i,N+1)\}]$
\item
$[i,t_{N+1},\theta_{h+\nu}\cup\{(i,N+1)\}] \vdash [i+1,t',\theta_h]
\vdash [i+1, p_{h+\nu+1}, \theta_{h+\nu+1}]$.
\end{itemize}
Since $p_h = p_{h+\nu}$ and $\A$ is deterministic,
we have $s_0=t_0$.
Our aim is to prove that $s_{N+1} = t_{N+1}$.
To this end, there are a few steps.
\begin{description}
\item[Step 1 (Application of the hypothesis that 
Claim~\ref{cl: proof subclaim B two} holds for the case $i$)]

For each $j \in \{0,\ldots,K(l+n_{i-1}+2)\}$,
we claim that $s_j = t_j$.

To apply the induction hypothesis that Claim~\ref{cl: proof subclaim B two}
for the case $i$,
we take the integers 
\begin{eqnarray*}
l_1 & = & l + n_{i-1} + 2
\\
l_2 & = & l + n_i +1
\\
l_3 & = & l_2
\\
l_4 & = & l + n_{i+1}
\end{eqnarray*}
Recall that $l$ is the integer such that
every pebble, except pebbles~$i$ and $(i+1)$,
are located either $\leq K(l)$, or $\geq L(l+n_{i+1})$.
Recall also that $\nu = \pi \beta_i !$.

It is straightforward to show that
$l_2-l_1 \geq n_{i-1}+1$ and $l_4 - l_3 \geq n_{i-1}+1$,
and all the conditions $(1)$, $(2)$ and $(3.a)$--$(3.h)$ hold.
Since $s_0 = t_0$, 
applying the hypothesis for each $j \in \{0,\ldots,K(l+n_{i-1}+2)\}$ 
-- that Claim~\ref{cl: proof subclaim B two} hold for the case $i$ --
we have $s_j = t_j$.

\item[Step 2 (Application of the hypothesis that
Claim~\ref{cl: proof subclaim B one} holds for the case $i$)]

For each $j \in \{K(l+n_{i-1}+1)+1,\ldots,K(l+n_{i-1}+2)-2\}$,
in the configuration $[i,s_j,\theta_h \cup \{(i,j)\}]$
the integer $l$ satisfies the condition that
each pebbles~$i+1,\ldots,k$ are located either $\leq K(l)$, or $L(l+n_i)$.

Applying the induction hypothesis that Claim~\ref{cl: proof subclaim B one} holds
for the case $i$,
there exist two integers $\nu_0'$ and $\nu'$ such that
\begin{itemize}
\item
$K(l+n_{i-1}+1)+1 \leq \nu_0' \leq K(l+n_{i-1}+1)+\beta_i$;
\item
$1 \leq \nu' \leq \beta_i$;
\item
$s_j = s_{j+\nu'}$, for each $j \in \{K(l+n_{i-1}+1)+1,\ldots,K(l+n_{i-1}+2)-\nu'-2\}$.
\end{itemize}
In particular, $\nu'$ divides $\beta_{i+1}$, by definition of $\beta_{i+1}$,
thus, 
$s_j = s_{j+\nu}$, for each $j \in \{K(l+n_{i-1}+1)+1,\ldots,K(l+n_{i-1}+2)-\nu-2\}$.

Similarly, we can show that
$t_j = t_{j+\nu}$, for each $j \in \{K(l+n_{i-1}+1)+1,\ldots,K(l+n_{i-1}+2)-\nu-2\}$.

\item[Step 3 (Application of the hypothesis that 
Claim~\ref{cl: proof subclaim B two} holds for the case $i$)]

For each $j \in \{K(l+n_{i-1}+1)+\nu_0,\ldots,L(n_i+2+n_{i-1}+1)\}$,
we claim that $s_j = t_{j+\nu}$.

To apply the induction hypothesis that 
Claim~\ref{cl: proof subclaim B two}
for the case $i$,
we take the following integers.
\begin{eqnarray*}
l_1 & = & l 
\\
l_2 & = & l + n_{i-1} +1
\\
l_3 & = & l + n_i + 2 + n_{i-1} + 1
\\
l_4 & = & l + n_{i+1}
\end{eqnarray*}
It is straightforward to show that
$l_2-l_1 \geq n_{i-1}+1$ and $l_4 - l_3 \geq n_{i-1}+1$,
and all the conditions $(1)$, $(2)$ and $(3.a)$--$(3.h)$ hold.

From Steps~1 and~2, we already have 
{\large
\begin{eqnarray*}
s_{K(l+n_{i-1}+1)+\nu_0} & = & t_{K(l+n_{i-1}+1)+\nu_0}
\\
s_{K(l+n_{i-1}+1)+\nu_0+\nu} & = & t_{K(l+n_{i-1}+1)+\nu_0+\nu}
\\
s_{K(l+n_{i-1}+1)+\nu_0} & = & s_{K(l+n_{i-1}+1)+\nu_0+\nu}
\end{eqnarray*}}
Applying the hypothesis for each $j \in \{K(l+n_{i-1}+1)+\nu_0,\ldots,L(n_i+2+n_{i-1}+1)-\nu\}$ 
-- that Claim~\ref{cl: proof subclaim B two} hold for the case $i$ --
on the configurations $[i,s_j,\theta_h \cup \{(i,j)\}]$ and
$[i,t_{j+\nu},\theta_h \cup \{(i,j+\nu)\}]$,
we have $s_j = t_{j+\nu}$.

\item[Step 4 (Application of the hypothesis that
Claim~\ref{cl: proof subclaim B one} holds for the case $i$)]

For each $j \in \{K(l+n_i + 2+ n_{i-1}+1)+1,\ldots,K(l+n_i + 2+ n_{i-1}+2)-2\}$,
in the configuration $[i,s_j,\theta_h \cup \{(i,j)\}]$
the integer $l+n_i+2$ satisfies the condition that
each pebbles~$i+1,\ldots,k$ are located either 
$\leq K(l+n_i+2)$, or $\geq L(l+n_i+2+n_i) = L(l+n_{i+1})$.

Applying the induction hypothesis that Claim~\ref{cl: proof subclaim B one} holds
for the case $i$,
there exist two integers $\nu_0''$ and $\nu''$ such that
\begin{itemize}
\item
$K(l+n_i + 2+ n_{i-1}+1)+1 \leq \nu_0'' \leq K(l+n_i + 2+ n_{i-1}+1)+\beta_i$;
\item
$1 \leq \nu'' \leq \beta_i$;
\item
$s_j = s_{j+\nu''}$, for each 
$j \in \{K(l+n_i + 2+ n_{i-1}+1)+1,\ldots,K(l+n_i + 2+ n_{i-1}+1)-\nu''-2\}$.
\end{itemize}
In particular, $\nu''$ divides $\beta_{i+1}$, and by definition of $\beta_{i+1}$,
thus, 
$s_j = s_{j+\nu}$, for each $j \in \{K(l+n_i + 2+ n_{i-1}+1)+1,\ldots,K(l+n_i + 2+ n_{i-1}+2)-\nu-2\}$.

Similarly, we can show that
$t_j = t_{j+\nu}$, for each $j \in \{K(l+n_i + 2+ n_{i-1}+1)+1,\ldots,K(l+n_i + 2+ n_{i-1}+2)-\nu-2\}$.
In particular, we have 
{\large
$$
s_{K(l+n_i + 2+ n_{i-1}+2)-2} = t_{K(l+n_i + 2+ n_{i-1}+2)-2}.
$$
}
By definition of $L(\cdot)$ and $K(\cdot)$,
this is equivalent to stating that
{\large
$$
s_{L(l+n_i + 2+ n_{i-1}+1)} = t_{L(l+n_i + 2+ n_{i-1}+1)}.
$$
}

\item[Step 5 (Application of the hypothesis that 
Claim~\ref{cl: proof subclaim B two} holds for the case $i$)]

For each $j \in \{L(l+n_{i-1}+1),\ldots,N+1\}$,
we claim that $s_j = t_{j}$.

To apply the induction hypothesis that 
Claim~\ref{cl: proof subclaim B two}
for the case $i$,
we take the integers 
\begin{eqnarray*}
l_1 & = & l 
\\
l_2 & = & l + n_{i-1} +1
\\
l_3 & = & l_2 
\\
l_4 & = & l_3 + n_{i-1} +1
\end{eqnarray*}
It is straightforward to show that
$l_2-l_1 \geq n_{i-1}+1$ and $l_4 - l_3 \geq n_{i-1}+1$,
and all the conditions $(1)$, $(2)$ and $(3.a)$--$(3.h)$ hold.

By Step~4, we already have
$s_{L(l+n_i + 2+ n_{i-1}+1)} = t_{L(l+n_i + 2+ n_{i-1}+1)}$.
Applying the hypothesis for each $j \in \{L(l+n_{i-1}+1),\ldots,N+1\}$
-- that Claim~\ref{cl: proof subclaim B two} hold for the case $i$ --
on the configurations $[i,s_j,\theta_h \cup \{(i,j)\}]$ and
$[i,t_{j},\theta_{h} \cup \{(i,j)\}]$,
we have $s_j = t_{j}$.
\end{description}
From here, as $s_{N+1}=t_{N+1}$ and $\A$ is deterministic,
we have $s'=t'$.
And again, by the deterministism of $\A$,
this implies $p_{h+1} = p_{h+1+\nu}$.
This completes the induction step for Claim~\ref{cl: proof subclaim B one}.

\paragraph*{\bf The induction step for Claim~\ref{cl: proof subclaim B two}}
We are going to show that Claim~\ref{cl: proof subclaim B two} holds for the case $i+1$.

Suppose $[i+1,q,\theta_1]$ and $[i+1,q,\theta_2]$
are configurations of $\A$ on $w(n_k,m)$
such that the conditions $(1)$, $(2)$, $(3.a)$-$(3.g)$ above hold.

Consider the following run:
\begin{itemize}
\item
$[i+1,q,\theta_1] \vdash [i,s_0,\theta_h\cup\{(i,0)\}]$;
\item
$[i,s_0,\theta_h\cup\{(i,0)\}] \vdash^{\ast}
\cdots \vdash^{\ast}
[i,s_{N+1},\theta_1\cup\{(i,N+1)\}]$
\item
$[i,s_{N+1},\theta_1\cup\{(i,N+1)\}] \vdash [i+1,s',\theta_1]
\vdash [i+1, p, Succ_{i+1}(\theta_1)]$.
\end{itemize}
and the following run:
\begin{itemize}
\item
$[i+1,q,\theta_2] \vdash [i,t_0,\theta_2\cup\{(i,0)\}]$;
\item
$[i,t_0,\theta_2\cup\{(i,0)\}] \vdash^{\ast}
\cdots \vdash^{\ast}
[i,t_{N+1},\theta_2\cup\{(i,N+1)\}]$
\item
$[i,t_{N+1},\theta_2\cup\{(i,N+1)\}] \vdash [i+1,t',\theta_2]
\vdash [i+1, r, Succ_{i+1}(\theta_2)]$.
\end{itemize}
We are going to show that $p=r$.
It can be proved in a similar manner
as in the proof of the induction step of Claim~\ref{cl: proof subclaim B one},
thus, omitted.

Briefly, the proof is divided into the same Steps~1--5 above.
The reasoning on each step still applies in this induction step,
and at the end we obtain $s_{N+1}=t_{N+1}$,
thus, $s'=t'$ and $p=r$.

\section{Weak PA}
\label{s: weak PA}

There is an analogue of our results from the previous section to another, 
but weaker, version of pebble automata. 
In the model defined in Section~\ref{s: model}, 
the new pebble is placed in the beginning of the input word. 
This model is called {\em strong} PA in~\cite{NevenSV04}.
An alternative would be to place the new pebble
at the position of the most recent one. The model defined this way
is usually referred as {\em weak} PA. Formally, it is defined by
setting $\theta'(i-1)=\theta(i)$ (and keeping
$\theta'(i)=\theta(i)$) in the case of $\ttAct = \ttPlace$ in the
definition of the transition relation in Definition~\ref{d: pa}.

We give the formal definition below.
\begin{definition}
\label{d: weak pa}
A {\em two-way alternating} weak $k$-{\em pebble automaton}, 
(in short weak $k$-PA) is a system $\A = \langle Q, q_0, F, \mu, U\rangle$ whose
components are defined as follows.
\begin{enumerate}
\item 
$Q$, $q_0 \in Q$ and $F\subseteq Q$ are a finite set of 
{\em states}, the {\em initial state}, and the set of {\em final
states}, respectively; 
\item 
$U \subseteq Q - F$ is the set of
{\em universal} states; and 
\item 
$\mu$ is a finite set of
transitions of the form $\alpha \rightarrow \beta$ such that
\begin{itemize}
\item 
$\alpha$ is of the form $(i,P,V,q)$,
where $i \in \{1,\ldots,k\}$, $P,V \subseteq\{i+1,\ldots,k\}$,
$q \in Q$ 
and 
\item $\beta$ is of the form $(q,\ttAct)$,
where $q \in Q$ and
\[
\ttAct\in \{ \ttRight,\ttPlace,\ttLift \}.
\]
\end{itemize}
\end{enumerate}
\end{definition}

The definitions of pebble assignment, configurations, initial and final configurations
as well as application of a transition on configurations are
the same as defined in the case of strong PA in Subsection~\ref{ss: pa}.

We define the transition relation $\vdash_{\sA}$ on $\triangleleft w \triangleright$ as follows:
$[i,q,\theta]\vdash_{\sA,w} [i^\prime,q^\prime,\theta^\prime]$, if
there is a transition $\alpha \rightarrow (p,\ttAct) \in \mu$ that
applies to $[i,q,\theta]$ such that $q^\prime = p$, 
for all $j > i$, $\theta^\prime(j)=\theta(j)$, and
\begin{itemize}
\item 
if $\ttAct = \ttRight$,
then $i^\prime=i$ and $\theta^\prime(i)=\theta(i)+1$, 
\item 
if $\ttAct = \ttLift$, then $i^\prime=i+1$, 
\item 
if $\ttAct = \ttPlace$, then $i^\prime=i-1$, $\theta^\prime(i-1) = \theta(i)$ and
$\theta^\prime(i)=\theta(i)$.
\end{itemize}
Note the difference on the definition of $\theta'$
for the case of $\ttAct = \ttPlace$ from the one 
in the case of strong PA in Subsection~\ref{ss: pa}.

\begin{theorem}~\cite[Theorem~3]{Tan10}
\label{t: equivalence weak PA}
For each $k \geq 1$,
one-way alternating, nondeterministic and deterministic weak $k$-PA have
the same recognition power.
\end{theorem}

However, weak $k$-PA is weaker than strong $k$-PA.
For example, $\cR_{2^k-1}$ is not a weak $k$-PA language, 
see Lemma~\ref{l: R_k for weak PA} below.

Let
\[
\textrm{wPA}_k = \{L \mid L \textrm{ is accepted by a weak }k\textrm{-PA}\}
\]
and
\[
\textrm{wPA} = \bigcup_{k \geq 1} \textrm{wPA}_k
\]

The following lemma is the weak PA version of Proposition~\ref{p: savitch} and 
Corollary~\ref{c: R_n_k not in PA_k}.

\begin{lemma}
\label{l: R_k for weak PA}
For each $k=1,2,\ldots$, $\cR_k^+ \in
\textrm{wPA}_k$, but $\cR_{k+1}^+ \notin \textrm{wPA}_{k}$.
\end{lemma}
\begin{proof}
First, we prove that $\cR_{k}^+ \in \textrm{wPA}_k$.
The weak $k$-PA $\A$ that accepts $\cR_k^+$ works as follows.
On an input word $w = a_0b_0\cdots a_nb_n$, 
it works as follows.
\begin{enumerate}
\item
It places pebble~$k$ on the second position to read the symbol $b_0$.
\item
For each $i=k-1,\ldots,1$, it does the following.
\begin{enumerate}
\item
Place pebble~$i$, and non-deterministically moves it right
until it finds an odd position that contains the same symbol read by pebble~$i+1$.
\item
If it finds such position, it moves pebble~$i$ one step to the right.
\item
If it cannot find such position, it rejects the input word.
\end{enumerate} 
\item
If at the end, pebble~$1$ is on the last position, then 
the automaton accepts the input word.
\end{enumerate}
It is quite straightforward to
show that the automaton $\A_k$ accepts $\cR_k^+$.

Now we prove that $\cR_{k+1}^{+}\notin\textrm{wPA}_k$.
Suppose to the contrary that there is a weak $k$-PA $\A$
that accepts $\cR_{k+1}^+$.
By adding some extra states,
we can normalise the behaviour of each pebble as follows.
For each $i\in \{1,\ldots,k\}$, pebble~$i$ behaves as follows.
\begin{itemize}
\item
After pebble~$i$ moves right, then pebble~$(i-1)$ (when $i> 1$) is immediately placed
(in position 0 reading the left end-marker $\triangleleft$).
\item
If $i < k$, pebble~$i$ is lifted only when
it reaches the right-end marker $\triangleright$
of the input.
\item
Immediately after pebble~$i$ is lifted, pebble~$(i+1)$ moves right.
\end{itemize}
We also assume that in the automaton $\A$ 
only pebble~$k$ can enter a final state and
it may do so only after it reads the right-end marker $\triangleright$
of the input.

We let $m = \beta_{k+1}$, as defined in Subsection~\ref{ss: proof},
where
$\beta_0 = 1$, $\beta_1 = |Q|$, and for $i \geq 2$,
\[
\beta_{i}  =  |Q|! \times \beta_{i-1}!
\]
Also recall that the words $w(k+1,m)$ and $\overline{w}(k+1,m)$ are defined as follows.
\begin{eqnarray*}
w(k+1,m)  & = & 
a_0a_1 C_1 b_0b_1  D_1  
\cdots\cdots 
a_{k-1}a_{k}  C_{k}  b_{k-1}b_{k}  D_{k} a_{k}a_{k+1}
\\
\overline{w}(k+1,m)  & = & 
a_0a_1 C_1 b_0b_1  D_1  
\cdots\cdots 
a_{k-1}a_{k}  C_{k}  b_{k-1}b_{k},
\end{eqnarray*}
where for each $i=1,\ldots,k$,
\begin{itemize}
\item
$C_i = c_{i,1}c_{i+1,1}\ \cdots \ c_{i,m-1}c_{i+1,m-1} $;
\item
$D_i = d_{i,1}d_{i+1,1}\ \cdots \ d_{i,m-1}d_{i+1,m-1} $.
\end{itemize}

Obviously $w(k+1,m) \in \cR_{k+1}^+$, while $\overline{w}(k+1,m) \notin \cR^+$.
We establish the following claim that immediately implies $\cR_{k+1}^+ \notin \textrm{wPA}_k$.
\begin{claim}
The automaton $\A$ either accepts both $w(k+1,m)$ and $\overline{w}(k+1,m)$,
or rejects both $w(k+1,m)$ and $\overline{w}(k+1,m)$.
\end{claim}
\begin{proof}
The proof is similar to the proof of Proposition~\ref{p: accept or reject both}.
So we simply sketch it here.
Let
$$
[k,p_0,\overline{\theta}_0] \
\vdash^{\ast}_{\sA,w(n_k,m)} \ \cdots\cdots \ \vdash^{\ast}_{\sA,w(k,m)} \
[k,p_{N+1},\overline{\theta}_{N+1}]
$$
be a run of $\A$ on $w(k+1,m)$, 
where $N$ is the length of $w(k+1,m)$ and
$\theta_j(k)=j$, for each $j \in \{0,\ldots,N+1\}$.

Let
$$
[k,r_0,\overline{\theta}_0] \
\vdash^{\ast}_{\sA,\overline{w}(k+1,m)} \ \cdots\cdots \ \vdash^{\ast}_{\sA,\overline{w}(k+1,m)} \
[k,r_{M+1},\overline{\theta}_{M+1}]
$$
be a run of $\A$ on $\overline{w}(k+1,m)$, 
where $M$ is the length of $\overline{w}(k+1,m)$ and
$\theta_j(k)=j$, for each $j \in \{0,\ldots,M+1\}$.

Now $p_0=r_0$, as both are the initial state of $\A$.
We are going to show that $p_{N+1}=r_{M+1}$.
It consists of three steps.
\begin{description}
\item[Step 1]
$p_{m} = r_{m}$.
\\
This step is similar to Claim~\ref{cl: main claim} proved in Subsection~\ref{ss: proof}.
That is, suppose $[k,q,\theta]$ and $[k,q,\overline{\theta}]$
are configurations on $w(k+1,m)$ and $\overline{w}(k+1,m)$, respectively,
and $0 \leq \theta(k)=\overline{\theta}(k) \leq m$.
If
\begin{eqnarray*}
~[k,q,\theta] & \vdash_{\sA,w(k+1,m)}^{\ast} & [k,p,Succ_{k}(\theta)]
\\
~[k,q,\overline{\theta}] & \vdash_{\sA,\overline{w}(k+1,m)}^{\ast} & [k,r,Succ_{k}(\theta)]
\end{eqnarray*}
then $p=r$.\footnote{The 
only difference between this proof and the proof of Claim~\ref{cl: main claim}
is that
here the induction hypothesis is that 
for each $1\leq i \leq k-1$,
weak $i$-PA cannot differentiate
between $w(i+1,m)$ and $\overline{w}(i+1,m)$;
while in Claim~\ref{cl: main claim} the induction hypothesis is
strong $i$-PA cannot differentiate between
$w(n_i,m)$ and $\overline{w}(n_i,m)$.}

\item[Step 2]
$r_{m} = p_{m} = p_{2m}$.
\\
This step is similar to Claim~\ref{cl: proof subclaim B one} stated in Subsection~\ref{ss: proof}.
That is, there exist two integers $\nu_0$ and $\nu$
such that for every $h \in \{m+\nu_0,\ldots,2m-\nu\}$,
we have $p_{h}=p_{h+\nu}$.
\\
The main idea is that since the integer $m$ is big enough,
there exists an integer $\nu$ such that on every $\nu$ steps,
pebble~$k$ will enter into the same state.
The integer $m$ is defined so that it is divisible by every possible such $\nu$,
thus, implies $p_{m}=p_{2m}$.
That $r_{m} = p_{m}$ is deduced from the previous step.

\item[Step 3]
$p_{N+1} = r_{M+1}$.
\\
Here we make use of the fact that $\A$ is a weak PA.
From previous step we have $p_{2m}=r_m$.
On the configuration $[k,p_{2m},\theta_{2m}]$ of $\A$ on $w(k+1,m)$,
pebble~$k$ only ``sees'' 
$a_1a_2 C_2 b_1b_2  D_2  \cdots a_{k-1}a_{k}  C_{k}  b_{k-1}b_{k}  D_{k} a_{k}a_{k+1}$;
while on the configuration $[k,r_{m},\theta_{m}]$ of $\A$ on $\overline{w}(k+1,m)$,
pebble~$k$ only ``sees'' 
$b_0b_1  D_1  \cdots a_{k-1}a_{k}  C_{k}  b_{k-1}b_{k}$.
\\
Since $a_1a_2 C_2 b_1b_2  D_2  \cdots a_{k-1}a_{k}  C_{k}  b_{k-1}b_{k}  D_{k} a_{k}a_{k+1}$
and $b_0b_1  D_1  \cdots a_{k-1}a_{k}  C_{k}  b_{k-1}b_{k}$
are essentially the same,
we have $p_{2m+1}=r_{m+1}$.
Similarly, from $p_{2m+1}=r_{m+1}$, we also can conclude 
that $p_{2m+2}=r_{m+2}$ and then $p_{2m+3}=r_{m+3}$ and so on until we get $p_{N+1}=r_{M+1}$.
\end{description}
\end{proof}
This completes the proof of Lemma~\ref{l: R_k for weak PA}.
$\qed$
\end{proof}

Lemma~\ref{l: R_k for weak PA} immediately implies 
the strict hierarchy for wPA languages.

\begin{theorem}
\label{t: weak PA} 
For each $k=1,2,\ldots$, $\textrm{wPA}_k
\subsetneq \textrm{wPA}_{k+1}$.
\end{theorem}

\section{Linear temporal logic with one register freeze quantifier}
\label{s: ltl}

In this section we recall the definition of Linear Temporal Logic
(LTL) augmented with one register freeze quantifier~\cite{DemriL09}.
We consider only one-way temporal operators ``next'' $\ttX$ and ``until'' $\ttU$, and do not
consider their past time counterparts.
Moreover, in~\cite{DemriL09} the LTL model is defined over data words.
Since in this paper we essentially ignore the finite labels,
the LTL model presented here also ignores the finite labels.
However, the result here can be adopted in a straightforward manner
for the data word model.

Roughly, the logic LTL$_1^{\downarrow}(\ttX,\ttU)$ is
the standard LTL augmented with a register to store a symbol from the infinite alphabet.
Formally, the formulas are defined as follows.
\begin{itemize}
\item
Both $\sfTrue$ and $\sfFalse$ belong to LTL$_1^{\downarrow}(\ttX,\ttU)$.
\item
$\uparrow$ is in LTL$_1^{\downarrow}(\ttX,\ttU)$.
\item
If $\varphi,\psi$ are in LTL$_1^{\downarrow}(\ttX,\ttU)$,
then so are $\neg \varphi$, $\varphi\vee\psi$ and $\varphi\wedge\psi$.
\item
If $\varphi$ is in LTL$_1^{\downarrow}(\ttX,\ttU)$,
then so is $\ttX \varphi$.
\item
If $\varphi$ is in LTL$_1^{\downarrow}(\ttX,\ttU)$,
then so is $\downarrow \varphi$.
\item
If $\varphi,\psi$ are in LTL$_1^{\downarrow}(\ttX,\ttU)$,
then so is $\varphi\ttU\psi$.
\end{itemize}
Intuitively, the predicate $\uparrow$ is intended to mean that the
current symbol is the same as the symbol in the register,
while $\downarrow \varphi$ is intended to mean that the formula
$\varphi$ holds when the register contains the current symbol.
This will be made precise in the definition of the semantics of
LTL$^{\downarrow}_1(\ttX,\ttU)$ below.

An occurrence of $\uparrow$ within the scope of some freeze
quantification $\downarrow$ is bounded by it; otherwise, it is
free. A sentence is a formula with no free occurrence of
$\uparrow$.

Next, we define the {\em freeze quantifier rank} of a sentence
$\varphi$, denoted by $\sffqr(\varphi)$.
\begin{itemize}
\item
$\sffqr(\sfTrue)=\sffqr(\sfFalse)=\sffqr(\uparrow)=0$.
\item
$\sffqr(\ttX\varphi)=\sffqr(\neg\varphi)=\sffqr(\varphi)$,
for every $\varphi$ in LTL$_1^{\downarrow}(\ttX,\ttU)$.
\item
$\sffqr(\varphi\vee\psi)=\sffqr(\varphi\wedge\psi)=\sffqr(\varphi\ttU\psi)=\max(\sffqr(\varphi),\sffqr(\psi))$,
for every $\varphi$ and $\psi$ in LTL$_1^{\downarrow}(\ttX,\ttU)$.
\item
$\sffqr(\downarrow\varphi)=\sffqr(\varphi)+1$,
for every $\varphi$ in LTL$_1^{\downarrow}(\ttX,\ttU)$.
\end{itemize}

Finally, we define the semantics of
LTL$_1^{\downarrow}(\ttX,\ttU)$.
Let $w=a_1 \cdots a_n$ be a word.
For a position $l = 1,\ldots,n$, a symbol $a$ and a formula $\varphi$ in
LTL$_1^{\downarrow}(\ttX,\ttU)$,
$w,l \models_{a} \varphi$ means that $\varphi$ is satisfied by $w$ at position $l$ when the
content of the register is $a$.
As usual, $w,l \not\models_{a} \varphi$
means the opposite.
The satisfaction relation is defined
inductively as follows.

\begin{itemize}
\item
$w,l \models_a \sfTrue$ and $w,l \not\models_a
\sfFalse$, for all $l=1,2,3,\ldots$ and $a\in\fD$.
\item
$w,l \models_a\; \varphi \vee \psi$ if and only if $w,l \models_a\;
\varphi$  or  $w,l \models_a \;\psi$.
\item
$w,l \models_a\; \varphi \wedge \psi$ if and only if $w,l \models_a\; \varphi$ and
$w,l \models_a \;\psi$.
\item
$w,l \models_a\; \neg \varphi$  if
and only if $w,l \not\models_a\; \varphi$.
\item $w,l \models_a\;
\ttX \varphi$  if and only if $1\leq l < n$  and 
$w,l+1 \models_a\; \varphi$.
\item
$w,l \models_a\; \varphi\ttU\psi$ if
and only if there exists $l'\geq l$ such that
$w,l' \models_a\; \psi$ and $w,l'' \models_a\; \varphi$,
for all $l''=i,\ldots,l'-1$.
\item
$w,l \models_a\; \downarrow\!\varphi$  if and only if 
$w,l \models_{a_l}\; \varphi$
\item
$w,l \models_a\; \uparrow$  if and only if $a = a_l$.
\end{itemize}
For a sentence $\varphi$ in
LTL$_1^{\downarrow}(\ttX,\ttU)$, 
we write $w,1 \models \varphi$, 
if $w,1 \models_{a} \varphi$ for some $a \in \fD$.
Note that since $\varphi$ is a sentence, all occurrences of
$\uparrow$ in $\varphi$ are bounded. Thus, it makes no difference
which data value $a$ is used in the statement $w,1 \models_a
\varphi$ of the definition of $w,1 \models \varphi$.
We define the language $L(\varphi)$ by
$L(\varphi) =  \{w \mid w,1 \models \varphi\}$.

\begin{theorem}
\label{t: LTL1 is in top view weak k-pa}
For every sentence $\psi
\in \textrm{LTL}^{\downarrow}(\ttX,\ttU)$,
there exists a weak $k$-PA $\A_{\psi}$, where $k = \sffqr(\psi)+1$, such
that $L(\A_{\psi})=L(\psi)$.
\end{theorem}
\begin{proof}
Let $\psi$ be an $\textrm{LTL}^{\downarrow}_1(\ttX,\ttU)$
sentence. We construct an alternating weak $k$-PA
$\A_{\psi}$, where $k=\sffqr(\psi)+1$ such that given a word
$w$, the automaton $\A_{\psi}$ ``checks'' whether $w,1\models
\psi$. $\A_{\psi}$ accepts $w$ if it is so. Otherwise, it rejects.

Intuitively, the computation of $w,1\models\psi$ is done
recursively as follows. The automaton $\A_{\psi}$ ``consists of''
the automata $\A_{\varphi}$ for all sub-formula of $\psi$.
\begin{itemize}
\item
If $\psi = \varphi \vee \varphi'$, then $\A_{\psi}$
nondeterministically chooses one of $\A_{\varphi}$ or
$\A_{\varphi'}$ and proceeds to run one of them.
\item
If $\psi = \varphi \wedge \varphi'$, then $\A_{\psi}$ splits its computation
(by conjunctive branching) into two and proceeds to run both
$\A_{\varphi}$ and $\A_{\varphi'}$.
\item
If $\psi = \ttX \varphi$, $\A_{\psi}$ moves to the right one step. If it reads the
right-end marker $\triangleright$, then it rejects immediately.
Otherwise, it proceeds to run $\A_{\varphi}$.
\item
If $\psi = \uparrow$, then $\A_{\psi}$ checks whether the symbol
seen by its head pebble is the same as the one seen by the second
last placed pebble. If it is not the same, then it rejects
immediately.
\item
If $\psi = \downarrow\varphi$, then $\A_{\psi}$ places a new
pebble and proceeds to run $\A_{\varphi}$.
\item
If $\psi = \varphi \ttU \varphi'$, then 
$\A_{\psi}$ repeatedly does the following.
\begin{enumerate}
\item
It splits its computation (by conjunctive branching) into two.
\item
In one branch it runs $\A_{\varphi}$.
\item
In the other it moves one step to the right
and starts on Step~1 again.
\end{enumerate}
It repeatedly performs (1)--(3) until it nondeterministically decides to run $\A_{\varphi'}$.
\item
If $\psi = \neg \varphi$, then $\A_{\psi}$ runs the complement of $\A_{\varphi}$.
The complement of $\A_{\varphi}$ can be constructed by switching the 
accepting states into non-accepting states and the non-accepting states into
accepting states, as well as, switching the universal states into 
non-universal states and the non-universal states into universal states.
\end{itemize}
Note that since $\sffqr(\varphi)=k$, on each computation path the
automaton $\A_{\psi}$ only needs to place the pebble $k$ times,
thus, $\A_{\psi}$ requires only $(k+1)$ pebbles.

Now it is a straightforward induction on the length of $\varphi$ 
to show that 
$$
w,l \models_a \varphi
\ \mbox{if and only if} \
\mbox{the configuration} \
[i,q,\theta] \
\mbox{leads to acceptance},
$$
where 
\begin{itemize}
\item
$i = \sffqr(\varphi)+1$;
\item
$q$ is the initial state of $\A_{\varphi}$;
\item
$\theta$ is a pebble assignment where
$\theta(i)=l$ and $\theta(j) \leq l$,
for each $j \in \{i+1,\ldots,k+1\}$;
\item
$a$ is the symbol seen by pebble~$(i+1)$, if $i \neq k+1$.
(If $i = k+1$, then $a$ can be an arbitrary symbol.)
\end{itemize}
From here, it immediately follows that $L(\A_{\psi})=L(\psi)$.
\end{proof}

Our next results deal with the expressive power of
LTL$^{\downarrow}_1(\ttX,\ttU)$ based on the freeze
quantifier rank. It is an analog of the classical hierarchy of
first order logic based on the ordinary quantifier rank. We start
by defining an LTL$^{\downarrow}_1(\ttX,\ttU)$ sentence for
the language $\cR^+_m$ defined in Section~\ref{s: graph}.

\begin{lemma}
\label{l: R_k for LTL}
For each $k=1,2,3,\ldots$, there exists a
sentence $\psi_{k}$ in LTL$^{\downarrow}_1(\ttX,\ttU)$ such
that $L(\psi_{k})=\cR^+_{k}$ and
$\sffqr(\psi_1)=1$; and $\sffqr(\psi_{k})=k-1$, when
$k\geq 2$.
\end{lemma}
\begin{proof}
First, we define a formula $\varphi_k$ such that
$\sffqr(\varphi_k)=k-1$ and for every word $w = d_1 \cdots d_n$, for every
$i=1,\ldots,n$,
\begin{eqnarray}
\label{eq: phi_k for R_k^+}
w,i \models_{d_i} \varphi_k &
\textrm{if and only if} & d_i\cdots  d_n \in \cR^+_{k}.
\end{eqnarray}
We construct $\varphi_k$ inductively as follows.
\begin{itemize}
\item
$\varphi_1 = \ttX(\neg \uparrow) \wedge \neg(\ttX(\ttX\ \sfTrue))$.
\item
For each $k=1,2,3,\ldots$,
\begin{eqnarray*}
\varphi_{k+1} & = & \ttX(\neg \uparrow) \wedge \ttX \Big(
\downarrow \ttX \Big( (\neg\uparrow) \ttU (\uparrow \wedge
\varphi_k) \Big) \Big)
\end{eqnarray*}
\end{itemize}
Note that since $\sffqr(\varphi_1)=0$, then for each
$k=1,2,\ldots$, $\sffqr(\varphi_k)=k-1$.

It is straightforward to show that $\varphi_k$ satisfies Equation~(\ref{eq: phi_k for R_k^+}).
The desired sentence $\psi_k$ is defined as follows.
\begin{itemize}
\item
$\psi_1 = \downarrow \big(\ttX(\neg \uparrow) \wedge
\neg(\ttX(\ttX\ \sfTrue))\big)$.
\item
For each $k=2,3,\ldots$,
\begin{eqnarray*}
\psi_{k} & = & \downarrow(\ttX(\neg \uparrow)) \wedge \ttX \Big(
\downarrow \ttX \Big( (\neg\uparrow) \ttU (\uparrow \wedge
\varphi_{k-1}) \Big) \Big)
\end{eqnarray*}
\end{itemize}
Obviously, $\sffqr(\psi_1)=1$. For $k \geq 2$,
$\sffqr(\varphi_{k-1})=k-2$, thus, $\sffqr(\psi_k)=k-1$.
\end{proof}

\begin{lemma}
\label{l: R_k+1 not in fqr k-1}
For each $k=1,2,\ldots$, the language $\cR^+_{k+1}$ is not expressible by a sentence in
LTL$^{\downarrow}_1(\ttX,\ttU)$ of freeze quantifier rank
$(k-1)$.
\end{lemma}
\begin{proof}
By Lemma~\ref{l: R_k for weak PA}, $\cR^+_{k+1}$ is not
accepted by weak $k$-PA. 
Then, by Theorem~\ref{t: LTL1 is in top view weak k-pa},
$\cR^+_{k+1}$ is not expressible by
LTL$^{\downarrow}_1(\ttX,\ttU)$ sentence of freeze
quantifier rank $(k-1)$.
\end{proof}

Combining both Lemmas~\ref{l: R_k for LTL} and~\ref{l: R_k+1 not in fqr k-1},
we obtain that for each $k=1,2,\ldots$, the language
$\cR_{k+1}$ separates the class of
LTL$_1^{\downarrow}(\ttX,\ttU)$ sentences of freeze
quantifier rank $k$ from the class of
LTL$_1^{\downarrow}(\ttX,\ttU)$ sentences of freeze
quantifier rank $(k-1)$. Formally, we state it as follows.

\begin{theorem}
\label{t: fqr k+1 > fqr k}
For each $k=1,2,\ldots$, the class of
sentences in LTL$^{\downarrow}_1(\ttX,\ttU)$ of freeze
quantifier rank $k$ is strictly more expressive than those of
freeze quantifier rank $(k-1)$.
\end{theorem}

\vspace{0.5 cm}
\noindent
{\bf Acknowledgement.}
The author would like to thank the anonymous referees, both the conference and the journal versions,
for their careful reading and comments.
The author also would like to thank Michael Kaminski for his support and guidance
when this work was done.
\vspace{0.5 cm}

\bibliography{reference}
\bibliographystyle{acmtrans}

\end{document}